\documentclass[a4paper]{llncs}
\usepackage{amssymb,amsmath,graphicx,epstopdf}
\usepackage{enumerate}
\usepackage{framed}
\usepackage{ntheorem}

\usepackage{tikz}
\usetikzlibrary{positioning}

\newtheorem*{observation}{Observation.}
\newtheorem*{CSPdichotomy}{CSP dichotomy conjecture.}
\newtheorem*{algdichotomy}{Algebraic CSP dichotomy conjecture.}

{
\theoremheaderfont{\normalfont\itshape}
\theorembodyfont{\normalfont}
\theoremseparator{.}

}

\newcommand{\up}{\textup}
\newcommand{\CSP}{\operatorname{CSP}}
\newcommand{\G}{\mathbb{G}}
\newcommand{\A}{\mathbb{A}}
\newcommand{\B}{\mathbb{B}}
\renewcommand{\H}{\mathbb{H}} 
\newcommand{\Pp}{\mathbb{P}}
\newcommand{\Q}{\mathbb{Q}}
\newcommand{\hgt}[1]{\operatorname{hgt}(#1)}

\DeclareMathOperator{\Pol}{Pol}
\DeclareMathOperator{\lvl}{lvl}

\begin{document}

\author{Jakub Bul{\' i}n\inst{1} \and Dejan Deli\'c\inst{2}
    \and Marcel Jackson\inst{3} \and Todd Niven\inst{3}
    \\
    \vspace{3mm}
    \email{jakub.bulin@gmail.com, ddelic@ryerson.ca,
    m.g.jackson@latrobe.edu.au, toddniven@gmail.com}}

\institute{$^1$ Faculty of Mathematics and Physics,\\ Charles University in Prague, Czech Republic\\
           $^2$ Department of Mathematics, Ryerson University,  Canada\\
           $^3$ Department of Mathematics, La Trobe University, Australia}
\title{On the reduction of the CSP dichotomy conjecture to digraphs}
\maketitle

\begin{abstract} 
  It is well known that the constraint satisfaction problem over
  general relational structures can be reduced in polynomial time to
  digraphs.  We present a simple variant of such a reduction and use
  it to show that the algebraic dichotomy conjecture is equivalent to
  its restriction to digraphs and that the polynomial reduction can be
  made in logspace.  We also show that our reduction preserves the
  bounded width property, i.e., solvability by local consistency
  methods.  We discuss further algorithmic properties that are
  preserved and related open problems.

\let\thefootnote\relax\footnote{The
    first author was supported by the grant projects GA\v CR
    201/09/H012, GA UK 67410, SVV-2013-267317; the second author
    gratefully acknowledges support by the Natural Sciences and
    Engineering Research Council of Canada in the form of a Discovery
    Grant; the third and fourth were supported by ARC Discovery
    Project DP1094578; the first and fourth authors were also
    supported by the Fields Institute.}
\end{abstract}

\section{Introduction}
A fundamental problem in constraint programming is to understand the
computational complexity of constraint satisfaction problems
(CSPs). While it is well known that the class of all constraint
problems is NP-complete, there are many subclasses of problems for
which there are efficient solving methods. One way to restrict the
instances is to only allow a fixed set of constraint relations, often
referred to as a \emph{constraint language}~\cite{b-j-k} or
\emph{fixed template}. Classifying the computational complexity of
fixed template CSPs has been a major focus in the theoretical study of
constraint satisfaction. In particular it is of interest to know which
templates produce polynomial time solvable problems to help provide
more efficient solution techniques.

The study of fixed template CSPs dates back to the 1970's with the
work of Montanari~\cite{Montanari} and Schaefer~\cite{sch}. A standout
result from this era is that of Schaefer who showed that the CSPs
arising from constraint languages over 2-element domains satisfy a
\emph{dichotomy}. The decision problem for fixed template CSPs over
finite domains belong to the class NP, and Schaefer showed that in the
2-element domain case, a constraint language is either solvable in
polynomial time or NP-complete. Dichotomies cannot be expected for
decision problems in general, since (under the assumption that
P$\neq$NP) there are many problems in NP that are neither solvable in
polynomial time, nor NP-complete \cite{lad}. Another important
dichotomy was proved by Hell and Ne{\v s}et{\v
  r}il~\cite{helnes:1}. They showed that if a fixed template is a
finite simple graph (the vertices make up the domain and the edges
make up the only allowable constraints), then the corresponding CSP is
either polynomial time solvable or NP-complete. The decision problem
for a graph constraint language can be rephrased as graph homomorphism
problem (a graph homomorphism is a function from the vertices of one
graph to another such that the edges are preserved).  Specifically,
given a fixed graph $\mathcal H$ (the constraint language), an
instance is a graph $\mathcal G$ together with the question ``Is there
a graph homomorphism from $\mathcal G$ to $\mathcal H$?''. In this
sense, $3$-colorability corresponds to $\mathcal H$ being the complete
graph on $3$ vertices. The notion of graph homomorphism problems
naturally extends to directed graph (digraph) homomorphism problems
and to relational structure homomorphism problems.

These early examples of dichotomies, by Schaefer, Hell and Ne{\v
  s}et{\v r}il, form the basis of a larger project of classifying the
complexity of fixed template CSPs.  Of particular importance in this
project is to prove the so-called \emph{CSP Dichotomy Conjecture} of
Feder and Vardi~\cite{fedvar} dating back to 1993. It states that the
CSPs relating to a fixed constraint language over a finite domain are
either polynomial time solvable or NP-complete.  To date this
conjecture remains unanswered, but it has driven major advances in the
study of CSPs.
 
One such advance is the algebraic connection revealed by Jeavons,
Cohen and Gyssens~\cite{JCG97} and later refined by Bulatov, Jeavons
and Krokhin~\cite{b-j-k}. This connection associates with each finite
domain constraint language~$\mathbb{A}$ a finite algebraic structure
$\mathbf{A}$. The properties of this algebraic structure are deeply
linked with the computational complexity of the constraint
language. In particular, for a fixed constraint language~$\mathbb{A}$,
if there does not exist a particular kind of operation, known as a
Taylor polymorphism, then the class of problems determined
by~$\mathbb{A}$ is NP-complete. Bulatov, Jeavons and
Krokhin~\cite{b-j-k} go on to conjecture that all other constraint
languages over finite domains determine polynomial time CSPs (a
stronger form of the CSP Dichotomy Conjecture, since it describes
where the split between polynomial time and NP-completeness
lies). This conjecture is often referred to as the \emph{Algebraic CSP
  Dichotomy Conjecture}.  Many important results have been built upon
this algebraic connection.  Bulatov~\cite{bul3} extended
Schaefer's~\cite{sch} result on 2-element domains to prove the CSP
Dichotomy Conjecture for 3-element domains. Barto, Kozik and
Niven~\cite{b-k-n} extended Hell and Ne{\v s}et{\v r}il's
result~\cite{helnes:1} on simple graphs to constraint languages
consisting of a finite digraph with no sources and no sinks. Barto and
Kozik~\cite{b-k2} gave a complete algebraic description of the
constraint languages over finite domains that are solvable by local
consistency methods (these problems are said to be of \emph{bounded
  width}) and as a consequence it is decidable to determine whether a
constraint language can be solved by such methods.

In their seminal paper, Feder and Vardi~\cite{fedvar} not only
conjectured a dichotomy, they also reduced the problem of proving the
dichotomy conjecture to the particular case of digraph homomorphism
problems, and even to digraph homomorphism problems where the digraph
is balanced (here balanced means that its vertices can be partitioned
into levels).  Specifically, for every template $\mathbb{A}$ (a finite
relational structure of finite type) there is a balanced digraph
(digraphs are particular kinds of relational structures)
$\mathcal{D}(\mathbb A)$ such that the CSP over $\mathbb A$ is
polynomial time equivalent to that over $\mathcal{D}(\mathbb{A})$.

\section{The main results}\label{sec:mainresults}

In general, fixed template CSPs can be modelled as relational
structure homomorphism problems~\cite{fedvar}. For detailed formal
definitions of relational structures, homomorphisms and polymorphisms,
see Section~\ref{sec:definitions}. 

Let $\mathbb{A}$ be a finite structure with signature~$\mathcal{R}$
(the fixed template), then the \emph{constraint satisfaction problem
  for} $\mathbb{A}$ is the following decision problem.

\medskip
\noindent \textbf{Constraint satisfaction problem for $\mathbb A$.}\\
\fbox{\parbox{0.67\textwidth}{ $\mathbf{CSP}(\mathbb A)$ \hrule
    \medskip
    INSTANCE: A finite $\mathcal{R}$-structure $\mathbb{X}$.\\
    QUESTION: Is there a homomorphism from $\mathbb{X}$ to
    $\mathbb{A}$?}}  \medskip

\noindent The dichotomy conjecture~\cite{fedvar} can be stated as
follows:
\begin{CSPdichotomy}
  Let $\mathbb{A}$ be a finite relational structure. Then\linebreak
  $\CSP(\mathbb{A})$ is solvable in polynomial time or NP-complete.
\end{CSPdichotomy}
The dichotomy conjecture is equivalent to its restriction to
digraphs~\cite{fedvar}, and thus can be restated as follows:
\begin{CSPdichotomy}
  Let $\mathbb{H}$ be a finite digraph. Then $\CSP(\mathbb{H})$ is
  solvable in polynomial time or NP-complete.
\end{CSPdichotomy}
Every finite relational structure $\mathbb A$ has a unique \emph{core}
substructure $\mathbb A'$ (see Section~\ref{sec:rel struct} for the
precise definition) such that $\CSP(\mathbb{A})$ and
$\CSP(\mathbb{A}')$ are identical problems, i.e., the ``yes'' and
``no'' instances are precisely the same. The algebraic dichotomy
conjecture~\cite{b-j-k} is the following:
\begin{algdichotomy}
  Let $\mathbb{A}$ be a finite relational structure that is a core. If
  $\mathbb{A}$ has a Taylor polymorphism then $\CSP(\mathbb{A})$ is
  solvable in polynomial time, otherwise $\CSP(\mathbb{A})$ is
  NP-complete.
\end{algdichotomy} 
Indeed, perhaps the above conjecture should be called the
\emph{algebraic tractability conjecture} since it is known that if a
core $\mathbb A$ does not possess a Taylor polymorphism, then
$\CSP(\mathbb A)$ is NP-complete~\cite{b-j-k}.

Feder and Vardi~\cite{fedvar} proved that every fixed template CSP is
polynomial time equivalent to a digraph CSP. This article will provide
the following theorem, which replaces ``polynomial time'' with
``logspace'' and reduces the algebraic dichotomy conjecture to
digraphs.
\begin{theorem}\label{thm:1}
  Let $\mathbb{A}$ be a finite relational structure. There is a finite
  digraph~$\mathcal D(\mathbb A)$ such that
  \begin{enumerate}[(i)]
  \item $\CSP(\mathbb{A})$ and $\CSP(\mathcal D(\mathbb A))$ are
    logspace equivalent,
  \item $\mathbb{A}$ has a Taylor polymorphism if and only if
    $\mathcal D(\mathbb A)$ has a Taylor polymorphism, and
  \item $\mathbb{A}$ is a core if and only if $\mathcal D(\mathbb A)$
    is a core.
  \end{enumerate}
  Furthermore, if $\mathbb{A}$ is a core, then $\CSP(\mathbb{A})$ has
  bounded width if and only if $\CSP(\mathcal D(\mathbb A))$ has
  bounded width.
\end{theorem}

\begin{proof}
  To prove (i), one reduction follows from
  Lemma~\ref{lem:forwardreduction} and
  Lemma~\ref{lem:logspace_reduction}. The other reduction is
  Lemma~\ref{lem:reversereduction}.

  To prove (ii) we employ Theorem \ref{thm:Taylor_iff_WNU}; it
  suffices to show that $\mathbb A$ has a WNU polymorphism if and only
  if $\mathcal D(\mathbb A)$ has a WNU polymorphism. The forward
  implication 
  (which is the crucial part of our proof) 
  is proved in Lemma \ref{lem:preserves_WNUs} and the
  converse follows from Lemma~\ref{lem:forwardreduction} and
  Lemma~\ref{lem:ppdef_polymorphisms}.
  Item (iii) is Corollary~\ref{cor:core}.

  The preservation of bounded width follows from Corollary
  \ref{cor:core}, Lemma~\ref{lem:preserves_WNUs} and Theorem
  \ref{thm:many_WNUs}.\hfill$\qed$
\end{proof}
See Remark~\ref{rem:number} in Section~\ref{sec:reduction} for the
size of $\mathcal D(\mathbb A)$. The ``Taylor polymorphism'' in
Theorem~\ref{thm:1}~(ii) can be replaced by many other polymorphism
properties, but space constraints do not allow us to elaborate here.

As a direct consequence of Theorem~\ref{thm:1} (ii) and (iii) above,
we can restate the algebraic dichotomy conjecture:
\begin{algdichotomy}
  Let $\mathbb{H}$ be a finite digraph that is a core. If $\mathbb{H}$
  has a Taylor polymorphism then $\CSP(\mathbb{H})$ is solvable in
  polynomial time, otherwise $\CSP(\mathbb{H})$ is NP-complete.
\end{algdichotomy}

\section{Background and definitions}\label{sec:definitions}

We approach fixed template constraint satisfaction problems from the
``homomorphism problem'' point of view. For background on the
homomorphism approach to CSPs, see \cite{fedvar}, and for background
on the algebraic approach to CSP, see \cite{b-j-k}.

A \emph{relational signature} $\mathcal{R}$ is a (in our case finite)
set of \emph{relation symbols} $R_i$, each with an associated arity
$k_i$. A (finite) \emph{relational structure} $\mathbb{A}$ \emph{over
  relational signature} $\mathcal{R}$ (called an
\emph{$\mathcal{R}$-structure}) is a finite set $A$ (the
\emph{domain}) together with a relation $R_i\subseteq A^{k_i}$, for
each relation symbol $R_i$ of arity $k_i$ in $\mathcal{R}$. A
\emph{CSP template} is a fixed finite $\mathcal{R}$-structure, for
some signature $\mathcal{R}$.

For simplicity we do not distinguish the relation with its associated
relation symbol, however to avoid ambiguity, sometimes we write
$R^\mathbb A$ to indicate that $R$ belongs to $\mathbb A$. We will
often refer to the domain of relational structure $\mathbb A$ simply
by $A$. When referring to a fixed relational structure, we may simply
specify it as $\mathbb A = (A; R_1,R_2,\dots,R_k)$. For technical
reasons we require that all the relations of a relational structure
are nonempty.

\subsection{Notation}\label{sec:notation}
For a positive integer $n$ we denote the set $\{1,2,\dots,n\}$ by
$[n]$. We write tuples using boldface notation, e.g. $\mathbf
a=(a_1,a_2,\dots,a_k)\in A^k$ and when ranging over tuples we use
superscript notation,
e.g. $(\mathbf{r}^1,\mathbf{r}^2,\dots,\mathbf{r}^l)\in R^l\subseteq
(A^k)^l$, where $\mathbf{r}^i=(r^i_1,r^i_2,\dots,r^i_k)$, for
$i=1,\dots,l$.

Let $R_i\subseteq A^{k_i}$ be relations of arity $k_i$, for
$i=1,\dots,n$. Let $k=\sum_{i=1}^nk_i$ and $l_i=\sum_{j<i}k_j$. We
write $R_1\times\dots\times R_n$ to mean the $k$-ary relation
\[
\{ (a_1,\dots, a_k)\in A^k \mid (a_{l_i+1},\dots,a_{l_i+k_i})\in R_i
\text{ for } i=1,\dots,n\}.
\]

An \emph{$n$-ary operation} on a set $A$ is simply a mapping
$f:A^n\rightarrow A$; the number $n$ is the \emph{arity} of $f$.  Let
$f$ be an $n$-ary operation on $A$ and let $k>0$. We define $f^{(k)}$
to be the $n$-ary operation obtained by applying $f$ coordinatewise on
$A^k$. That is, we define the $n$-ary operation $f^{(k)}$ on $A^k$ by
\[
f^{(k)}(\mathbf a^1,\dots,\mathbf
a^n)=(f(a^1_1,\dots,a^n_1),\dots,f(a^1_k,\dots,a^n_k)),
\]
for $\mathbf a^1,\dots, \mathbf a^n\in A^k$.

We will be particularly interested in so-called idempotent operations.
An $n$-ary operation $f$ is said to be \emph{idempotent} if it
satisfies the equation
  \[
  f(x,x,\dots,x)=x.
  \]

\subsection{Homomorphisms, cores and polymorphisms}\label{sec:homs}
We begin with the notion of a relational structure homomorphism.
\begin{definition}\label{def:hom}
  Let $\mathbb A$ and $\mathbb B$ be relational structures in the same
  signature~$\mathcal{R}$. A \emph{homomorphism} from $\mathbb A$ to
  $\mathbb B$ is a mapping $\varphi$ from $A$ to $B$ such that for
  each $n$-ary relation symbol $R$ in $\mathcal{R}$ and each $n$-tuple
  $\mathbf{a}\in A^n$, if $\mathbf{a}\in R^\mathbb A$, then
  $\varphi(\mathbf{a})\in R^\mathbb B$, where $\varphi$ is applied to
  $\mathbf a$ coordinatewise.
\end{definition}

We write $\varphi:\mathbb A\to\mathbb B$ to mean that $\varphi$ is a
homomorphism from $\mathbb A$ to $\mathbb B$, and $\mathbb A\to\mathbb
B$ to mean that there exists a homomorphism from $\mathbb A$ to
$\mathbb B$. 

An \emph{isomorphism} is a bijective homomorphism $\varphi$ such that
$\varphi^{-1}$ is a homomorphism. A homomorphism $\mathbb A\to\mathbb
A$ is called an \emph{endomorphism}. An isomorphism from $\mathbb A$
to $\mathbb A$ is an \emph{automorphism}. It is an easy fact that if
$\mathbb A$ is finite, then every surjective endomorphism is an
automorphism.

A finite relational structure $\mathbb A'$ is a \emph{core} if every
endomorphism $\mathbb A'\to\mathbb A'$ is surjective (and therefore an
automorphism). For every $\mathbb A$ there exists a relational
structure $\mathbb A'$ such that $\mathbb A\to\mathbb A'$ and $\mathbb
A'\to\mathbb A$ and $\mathbb A'$ is minimal with respect to these
properties. The structure $\mathbb A'$ is called the \emph{core of
  $\mathbb A$}. The core of $\mathbb A$ is unique (up to isomorphism)
and $\CSP(\mathbb A)$ and $\CSP(\mathbb A')$ are the same decision
problems. Equivalently, the core of $\mathbb A$ can be defined as a
minimal induced substructure that $\mathbb A$ retracts
onto. (See~\cite{helnes} for details on cores for graphs, cores for
relational structures are a natural generalisation.)

The notion of \emph{polymorphism} is central in the so called
algebraic approach to $\CSP$. Polymorphisms are a natural
generalization of endomorphisms to higher arity operations.

\begin{definition}
  Given an $\mathcal{R}$-structure $\mathbb{A}$, an $n$-ary
  \emph{polymorphism} of $\mathbb{A}$ is an $n$-ary operation $f$ on
  $A$ such that $f$ preserves the relations of $\mathbb A$. That is,
  if $\mathbf{a}^1,\dots,\mathbf{a}^n\in R$, for some $k$-ary relation
  $R$ in $\mathcal{R}$, then $f^{(k)}(\mathbf a^1,\dots,\mathbf
  a^n)\in R$.  
\end{definition}
Thus, an endomorphism is a $1$-ary polymorphism.

In this paper we will be interested in the following kind of
polymorphisms.
\begin{definition}
  A \emph{weak near-unanimity} \textup(WNU\textup) polymorphism is an $n$-ary
  idempotent polymorphism $\omega$, for some $n\geq 3$, that satisfies
  the following equations \textup(for all $x,y$\textup{):}
  \[
  \omega(x,\dots,x,y)=\omega(x,\dots,x,y,x)=\dots=\omega(y,x,\dots,x).
  \]
  We call the above \emph{WNU equations}.
\end{definition}
Note that since we assume that a WNU polymorphism $\omega$ is
idempotent it also satisfies the equation
\[
\omega(x,x,\dots,x)=x.
\]
Of particular interest, with respect to the algebraic dichotomy
conjecture, are Taylor polymorphisms. We will not need to explicitly
define Taylor polymorphisms (and only need consider WNU polymorphisms)
by the following theorem.

\begin{theorem}{\rm \cite{maroti-mckenzie}} \label{thm:Taylor_iff_WNU}
  A finite relational structure $\mathbb A$ has a Taylor polymorphism
  if and only if $\mathbb A$ has a WNU polymorphism.
\end{theorem}
Weak near-unanimity polymorphisms can be also used to characterise
CSPs of bounded width (see \cite{b-k2} for a detailed explanation of
the bounded width algorithm).
\begin{theorem}{\rm \cite{b-k2,maroti-mckenzie}} \label{thm:many_WNUs}
  Let $\mathbb A$ be a finite relational structure that is a
  core. Then $\CSP(\mathbb A)$ is of bounded width if and only if
  $\mathbb A$ has WNU polymorphisms of all but finitely many arities.
\end{theorem}

\subsection{Primitive positive definability}\label{sec:rel struct}
A first order formula is called \emph{primitive positive} if it is an
existential conjunction of atomic formul\ae.  Since we only refer to
relational signatures, a primitive positive formula is simply an
existential conjunct of formul\ae\ of the form $x = y$ or
$(x_1,x_2,\dots,x_n) \in R$, where $R$ is a relation symbol of arity
$n$.

For example, if we have a binary relation symbol $E$ in our signature,
then the formula
\[
\psi(x,y) = (\exists z)((x,z)\in E\ \wedge\ (z,y)\in E),
\]
pp-defines a binary relation in which elements $a,b$ are related if
there is a directed path of length $2$ from $a$ to $b$ in $E$.

\begin{definition}
  A relational structure $\mathbb{B}$ is \emph{primitive positive
    definable} in $\mathbb{A}$ \textup(or~$\mathbb{A}$
  \emph{pp-defines} $\mathbb{B}$\textup) if
\begin{enumerate}[(i)]
\item the set $B$ is a subset of $A$ and is definable by a primitive
  positive formula interpreted in $\mathbb{A}$, and
\item each relation $R$ in the signature of $\mathbb{B}$ is definable
  on the set $B$ by a primitive positive formula interpreted in
  $\mathbb{A}$.
\end{enumerate}
\end{definition}

The following result relates the above definition to the complexity of
CSPs.
\begin{lemma}{\rm \cite{JCG97}}\label{lem:logspace_reduction}
  Let $\mathbb{A}$ be a finite relational structure that pp-defines
  $\mathbb{B}$. Then, $CSP(\mathbb{B})$ is polynomial time {\rm (}indeed,
  logspace{\rm )} reducible to $\CSP(\mathbb{A})$.
\end{lemma}
It so happens that, if $\mathbb{A}$ pp-defines $\mathbb{B}$, then
$\mathbb{B}$ inherits the polymorphisms of
$\mathbb{A}$. See~\cite{b-j-k} for a detailed explanation. 
\begin{lemma}{\rm \cite{b-j-k}}\label{lem:ppdef_polymorphisms}
  Let $\mathbb{A}$ be a finite relational structure that pp-defines
  $\mathbb{B}$. If $\varphi$ is a polymorphism of $\mathbb A$, then
  its restriction to $B$ is a polymorphism of $\mathbb B$.
\end{lemma}
In particular, if $\mathbb A$ pp-defines $\mathbb B$ and $\mathbb A$
has a WNU polymorphism $\omega$, then $\omega$ restricted to $B$ is a
WNU polymorphism of $\mathbb B$.

In the case that $\mathbb{A}$ pp-defines $\mathbb{B}$ and $\mathbb{B}$
pp-defines $\mathbb{A}$, we say that $\mathbb{A}$ and $\mathbb{B}$ are
\emph{pp-equivalent}. In this case, $\CSP(\mathbb{A})$ and
$\CSP(\mathbb{B})$ are essentially the same problems (they are
logspace equivalent) and $\mathbb{A}$ and $\mathbb{B}$ have the same
polymorphisms.

\begin{example}\label{example:1}
  Let $\mathbb{A}=(A;R_1,\dots,R_n)$, where each $R_i$ is $k_i$-ary,
  and define $R=R_1\times \dots \times R_n$. Then the structure
  $\mathbb{A}'=(A;R)$ is pp-equivalent to $\mathbb{A}$.  

  Indeed, let $k=\sum_{i=1}^nk_i$ be the arity of $R$ and
  $l_i=\sum_{j<i}k_j$ for $i=1,\dots,n$. The relation $R$ is
  pp-definable from $R_1,\dots,R_n$ using the formula
  \[
  \Psi(x_1,\dots,x_k)= \bigwedge_{i=1}^n
  (x_{l_i+1},\dots,x_{l_i+k_i})\in R_i.
  \]
  The relation $R_1$ can be defined from $R$ by the primitive positive
  formula
  \[
  \Psi(x_1,\dots,x_{k_1})=(\exists y_{k_1+1},\dots,\exists y_k)
  ((x_1,\dots,x_{k_1},y_{k_1+1},\dots,y_k)\in R)
  \]
  and the remaining $R_i$'s can be defined similarly.
\end{example}
Example \ref{example:1} shows that when proving Theorem \ref{thm:1} we can
restrict ourselves to relational structures with a single relation.

\subsection{Digraphs}
\noindent A \emph{directed graph}, or \emph{digraph}, is a relational
structure $\mathbb{G}$ with a single binary relation symbol $E$ as
its signature. We typically call the members of $G$ and $E^\mathbb{G}$
\emph{vertices} and \emph{edges}, respectively. We usually write $a\to
b$ to mean $(a,b)\in E^\mathbb{G}$, if there is no ambiguity.

A special case of relational structure homomorphism (see
Definition~\ref{def:hom}), is that of digraph homomorphism. That is,
given digraphs $\mathbb G$ and $\mathbb H$, a function $\varphi:G\to
H$ is a homomorphism if $(\varphi(a),\varphi(b))\in E^\mathbb{H}$
whenever $(a,b)\in E^\mathbb{G}$.

\begin{definition}
  For $i=1,\dots, n$, let $\mathbb{G}_i=(G_i,E_i)$ be digraphs. The
  \emph{direct product of} $\mathbb{G}_1,\dots,\mathbb{G}_n$, denoted
  by $\prod_{i=1}^n \mathbb{G}_i$, is the digraph with vertices
  $\prod_{i=1}^nG_i$ {\rm(}the cartesian product of the sets $G_i${\rm
    )} and edge relation
  \[
  \{(\mathbf{a},\mathbf{b})\in (\prod_{i=1}^nG_i)^2 \mid
  (a_i,b_i)\in E_i \text{ for } i=1\dots,n \}.
  \]
  If $\mathbb{G}_1=\dots = \mathbb{G}_n =\mathbb{G}$ then we write
  $\mathbb{G}^n$ to mean $\prod_{i=1}^n \mathbb{G}_i$.
\end{definition}

With the above definition in mind, an $n$-ary polymorphism on a
digraph $\mathbb{G}$ is simply a digraph homomorphism from
$\mathbb{G}^n$ to $\mathbb{G}$. 

\begin{definition}
  A digraph $\mathbb P$ is an \emph{oriented path} if it consists of a
  sequence of vertices $v_0,v_1,\dots,v_k$ such that precisely one of
  $(v_{i-1},v_i),(v_i,v_{i-1})$ is an edge, for each $i=1,\dots,k$.  We
  require oriented paths to have a direction; we denote the
  \emph{initial} vertex $v_0$ and the \emph{terminal} vertex $v_k$ by
  $\iota\mathbb P$ and $\tau\mathbb P$, respectively.
\end{definition}

Given a digraph $\mathbb G$ and an oriented path $\mathbb P$, we write
$a\stackrel{\mathbb{P}}{\longrightarrow} b$ to mean that we can walk
in $\mathbb G$ from $a$ following $\mathbb{P}$ to $b$, i.e., there
exists a homomorphism $\varphi:\mathbb P\to\mathbb G$ such that
$\varphi(\iota\mathbb P)=a$ and $\varphi(\tau\mathbb P)=b$. Note that
for every $\mathbb P$ there exists a primitive positive formula
$\psi(x,y)$ such that $a\stackrel{\mathbb{P}}{\longrightarrow} b$ if
and only if $\psi(a,b)$ is true in $\mathbb{G}$. If there exists an
oriented path $\mathbb P$ such that
$a\stackrel{\mathbb{P}}{\longrightarrow} b$, we say that $a$ and $b$
are \emph{connected}. If vertices $a$ and $b$ are connected, then the
\emph{distance} from $a$ to $b$ is the number of edges in the shortest
oriented path connecting them. Connectedness forms an equivalence
relation on $G$; its classes are called the \emph{connected
  components} of $\mathbb G$. We say that a digraph is connected if it
consists of a single connected component.

A connected digraph is \emph{balanced} if it admits a \emph{level
  function} $\lvl:G\to\mathbb N$, where
$\lvl(b)=\lvl(a)+1$ whenever $(a,b)$ is an edge, and
the minimum level is~$0$. The maximum level is called the
\emph{height} of the digraph. Oriented paths are natural examples of
balanced digraphs.

By a \emph{zigzag} we mean the oriented path
\mbox{$\bullet\boldsymbol\rightarrow\bullet\boldsymbol
  \leftarrow\bullet\boldsymbol\rightarrow\bullet$} and a \emph{single
  edge} is the path \mbox{$\bullet\boldsymbol\rightarrow\bullet$}. For
oriented paths $\mathbb P$ and $\mathbb P'$, the \emph{concatenation
  of $\mathbb P$ and $\mathbb P'$}, denoted by $\mathbb
P\dotplus\mathbb P'$, is the oriented path obtained by identifying
$\tau\mathbb P$ with $\iota\mathbb P'$.

Our digraph reduction as described in Section~\ref{sec:reduction}
relies on oriented paths obtained by concatenation of zigzags and
single edges. For example, the path in Figure~\ref{fig:path} is a
concatenation of a single edge followed by two zigzags and two more
single edges (for clarity, we organise its vertices into levels).
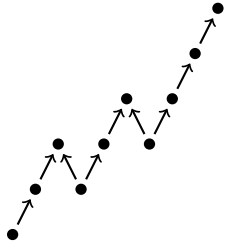
\begin{figure}[h]
\[
\begin{tikzpicture}[thick, on grid, node distance=.6cm and
  .3cm,dot/.style={circle,draw,outer sep=2.5pt,inner sep=0pt,minimum
    size=1.2mm,fill}]
\node [dot] (a) at (0,0) {};
\node [dot,above right=of a] (b)  {};
\node [dot,above right=of b] (c)  {};
\node [dot,below right=of c] (d)  {};
\node [dot,above right=of d] (e)  {};
\node [dot,above right=of e] (f)  {};
\node [dot,below right=of f] (g)  {};
\node [dot,above right=of g] (h)  {};
\node [dot,above right=of h] (i)  {};
\node [dot,above right=of i] (j)  {};
\draw[->] (a) -- (b);
\draw[->] (b) -- (c);
\draw[<-] (c) -- (d);
\draw[->] (d) -- (e);
\draw[->] (e) -- (f);
\draw[<-] (f) -- (g);
\draw[->] (g) -- (h);
\draw[->] (h) -- (i);
\draw[->] (i) -- (j);
\end{tikzpicture}
\]
\caption{A minimal oriented path}\label{fig:path}
\end{figure}

\section{The reduction to digraphs}\label{sec:reduction}

In this section we take an arbitrary finite relational structure
$\mathbb{A}$ and construct a balanced digraph $\mathcal D(\mathbb A)$
such that $\CSP(\mathbb{A})$ and $\CSP(\mathcal D(\mathbb A))$ are
logspace equivalent.

Let $\mathbb A=(A;R_1,\dots,R_n)$ be a finite relational structure,
where $R_i$ is of arity~$k_i$, for $i=1,\dots,n$. Let $k=\sum_{i=1}^n
k_i$ and let $R$ be the $k$-ary relation $R_1\times\dots\times
R_n$. For $\mathcal I\subseteq [k]$ define $\mathbb Q_{\mathcal I,l}$ to be a single
edge if $l\in\mathcal I$, and a zigzag if $l\in[k]\setminus\mathcal I$.

We define the oriented path
$\mathbb Q_\mathcal I$ (of height $k+2$) by
\[
\mathbb Q_\mathcal I\,=\, \mbox{$\bullet\boldsymbol\rightarrow\bullet$}
\dotplus\,\mathbb Q_{\mathcal I,1}\,\dotplus\,\mathbb
Q_{\mathcal I,2}\,\dotplus\,\dots\,\dotplus\,\mathbb Q_{\mathcal I,k}\,\dotplus\,
\mbox{$\bullet\boldsymbol\rightarrow\bullet$}
\]
Instead of $\mathbb Q_\emptyset,\mathbb Q_{\emptyset,l}$ we write just
$\mathbb Q,\mathbb Q_l$, respectively. For example, the oriented path
in Figure~\ref{fig:path} is $\mathbb Q_{\mathcal I}$ where $k=3$ and
$\mathcal I=\{3\}$. We will need the following observation.

\begin{observation}\label{obs:1}
  Let $\mathcal I,\mathcal J\subseteq [k]$. A homomorphism $\varphi:\mathbb
  Q_\mathcal I\to\mathbb Q_\mathcal J$ exists, if and only if $\mathcal I\subseteq\mathcal J$. In
  particular $\mathbb Q\to \mathbb Q_\mathcal I$ for all $\mathcal I\subseteq [k]$.
  Moreover, if $\varphi$ exists, it is unique and surjective.
\end{observation}

We are now ready to define the digraph $\mathcal D(\mathbb A)$.
\begin{definition}
  For every $e=(a,\mathbf r)\in A\times R$ we define $\mathbb{P}_e$ to
  be the path $\mathbb Q_{\{i:a=r_i\}}$.  The digraph $\mathcal
  D(\mathbb A)$ is obtained from the digraph $(A\cup R;A\times R)$ by
  replacing every $e=(a,\mathbf r)\in A\times R$ by the oriented path
  $\mathbb{P}_e$ \up(identifying $\iota\mathbb P_e$ with $a$ and
  $\tau\mathbb P_e$ with $\mathbf r$\up).
\end{definition}
(We often write $\mathbb P_{e,l}$ to mean $\mathbb Q_{\mathcal{I},l}$
where $\mathbb P_e=\mathbb Q_\mathcal{I}$.)

\begin{example}
  Consider the relational structure $\mathbb A=(\{ 0,1\}; R)$ where
  $R=\{ (0,1),(1,0)\}$, i.e., $\mathbb A$ is the
  $2$-cycle. Figure~\ref{fig:cycle} is a visual representation of
  $\mathcal D(\mathbb A)$.
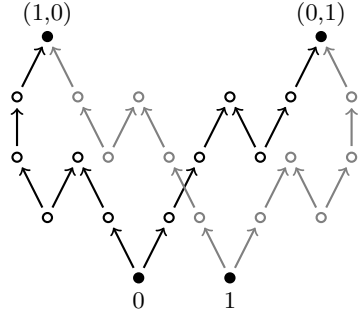
\begin{figure}[h]
\[
\begin{tikzpicture}[thick, on grid, node distance=.8cm and
  .4cm,
  dot/.style={circle,draw,outer sep=2.5pt,inner sep=0pt,minimum size=1.2mm,fill},
  emptydot/.style={circle,draw,outer sep=2.5pt,inner sep=0pt,minimum size=1.2mm,stroke},
  dot2/.style={circle,draw,outer sep=2.5pt,inner sep=0pt,minimum size=1.2mm,fill,color=gray},
  emptydot2/.style={circle,draw,outer sep=2.5pt,inner sep=0pt,minimum size=1.2mm,stroke,color=gray}
]
\node [dot] (a) at (-.6,0) {};
\node [below=3mm of a]{$0$};
\node [emptydot,above right=of a] (b)  {};
\node [emptydot,above right=of b] (e)  {};
\node [emptydot,above right=of e] (f)  {};
\node [emptydot,below right=of f] (g)  {};
\node [emptydot,above right=of g] (h)  {};
\node [dot,above right=of h] (i)  {};
\node [above=3mm of i]{(0,1)};

\draw[->] (a) -- (b);
\draw[->] (b) -- (e);
\draw[->] (e) -- (f);
\draw[<-] (f) -- (g);
\draw[->] (g) -- (h);
\draw[->] (h) -- (i);

\node [dot] (a2) at (.6,0) {};
\node [below=3mm of a2]{$1$};
\node [emptydot2,above left=of a2] (b2)  {};
\node [emptydot2,above left=of b2] (e2)  {};
\node [emptydot2,above left=of e2] (f2)  {};
\node [emptydot2,below left=of f2] (g2)  {};
\node [emptydot2,above left=of g2] (h2)  {};
\node [dot,above left=of h2] (i2)  {};
\node [above=3mm of i2]{(1,0)};

\draw[->,color=gray] (a2) -- (b2);
\draw[->,color=gray] (b2) -- (e2);
\draw[->,color=gray] (e2) -- (f2);
\draw[<-,color=gray] (f2) -- (g2);
\draw[->,color=gray] (g2) -- (h2);
\draw[->,color=gray] (h2) -- (i2);

\node [emptydot2,above right=of a2] (b2)  {};
\node [emptydot2,above right=of b2] (e2)  {};
\node [emptydot2,below right=of e2] (f2)  {};
\node [emptydot2,above right=of f2] (g2)  {};
\node [emptydot2,above =of g2] (h2)  {};

\draw[->,color=gray] (a2) -- (b2);
\draw[->,color=gray] (b2) -- (e2);
\draw[->,color=gray] (f2) -- (e2);
\draw[<-,color=gray] (g2) -- (f2);
\draw[->,color=gray] (g2) -- (h2);
\draw[->,color=gray] (h2) -- (i);

\node [emptydot,above left=of a] (b2)  {};
\node [emptydot,above left=of b2] (e2)  {};
\node [emptydot,below left=of e2] (f2)  {};
\node [emptydot,above left=of f2] (g2)  {};
\node [emptydot,above =of g2] (h2)  {};

 \draw[->] (a) -- (b2);
\draw[->] (b2) -- (e2);
\draw[->] (f2) -- (e2);
\draw[<-] (g2) -- (f2);
\draw[->] (g2) -- (h2);
\draw[->] (h2) -- (i2);

\end{tikzpicture}
\]
\caption{$\mathcal D(\mathbb A)$ where $\mathbb A$ is the
  $2$-cycle}\label{fig:cycle}
\end{figure}
\end{example}

\begin{remark}\label{rem:number}
  The number of vertices in $\mathcal D(\mathbb A)$ is
  $(3k+1)|R||A|+(1-2k)|R|+|A|$ and the number of edges is
  $(3k+2)|R||A|-2k|R|$. The construction of
  $\mathcal D(\mathbb A)$ can be performed in logspace \up(under any
  reasonable encoding\up).
\end{remark}
\begin{proof} The vertices of $\mathcal D(\mathbb A)$ consist of the
  elements of $A\cup R$, along with vertices from the connecting
  paths.  The number of vertices lying strictly within the connecting
  paths would be $(3k+1)|R||A|$ if every $\mathbb P_e$ was $\mathbb
  Q$. We need to deduct~$2$ vertices whenever there is a single edge
  instead of a zigzag and there are $\sum_{(a,\mathbf r)\in A\times
    R}|\{i:a=r_i\}|=k|R|$ such instances.  The number of edges is
  counted very similarly.\hfill$\qed$
\end{proof}

\begin{remark}
  Note that if we apply this construction to itself (that is,
  $\mathcal D(\mathcal D(\mathbb A))$) then we obtain balanced
  digraphs of height~$4$. When applied to digraphs, the $\mathcal{D}$
  construction is identical to that given by Feder and
  Vardi~\cite[Theorem 13]{fedvar}.
\end{remark}
The following lemma, together with Lemma~\ref{lem:logspace_reduction},
shows that $\CSP(\mathbb{A})$ reduces to $\CSP(\mathcal D(\mathbb A))$
in logspace.
\begin{lemma}\label{lem:forwardreduction}
$\mathbb{A}$ is pp-definable from $\mathcal D(\mathbb A)$.
\end{lemma}
\begin{proof}
  Example~\ref{example:1} demonstrates that $\mathbb A$ is
  pp-equivalent to $(A;R)$. We now show that $\mathcal D(\mathbb A)$
  pp-defines $(A;R)$, from which it follows that $\mathcal D(\mathbb
  A)$ pp-defines $\mathbb{A}$.

  Note that $\mathbb Q\to\mathbb P_e$ for all $e\in A\times R$, and
  $\mathbb Q_{\{i\}}\to\mathbb P_{(a,\mathbf r)}$ if and only if
  $a=r_i$.  The set $A$ is pp-definable in $\mathcal D(\mathbb A)$ by
  $A=\{x\mid (\exists y)(x\stackrel{\mathbb Q}{\longrightarrow} y)\}$
  and the relation~$R$ can be defined as the set
  $\{(x_1,\dots,x_k)\mid (\exists y)(x_i\stackrel{\mathbb
    Q_{\{i\}}}{\longrightarrow}y \text{ for all }i\in[k])\}$, which is
  also a primitive positive definition.\hfill$\qed$
\end{proof}

It is not, in general, possible to pp-define $\mathcal D(\mathbb A)$
from $\mathbb{A}$.\footnote[1]{Using the definition of pp-definability
  as described in this paper, this is true for cardinality
  reasons. However, a result of Kazda \cite{kaz} can be used to show
  that the statement remains true even for more general definitions of
  pp-definability.}  Nonetheless the following lemma is true.

\begin{lemma}\label{lem:reversereduction} 
  $\CSP(\mathcal D(\mathbb A))$ reduces in logspace to
  $\CSP(\mathbb{A})$.
\end{lemma}
The proof of Lemma \ref{lem:reversereduction} is rather technical,
though broadly follows the polynomial process described in the proof
of \cite[Theorem 13]{fedvar} (as mentioned, our construction coincides
with theirs in the case of digraphs).  Details of the argument will be
presented in a subsequent expanded version of this article.

Lemma \ref{lem:forwardreduction} and Lemma \ref{lem:reversereduction}
complete the proof of part (i) of Theorem \ref{thm:1}.  As this
improves the oft-mentioned polynomial time equivalence of general
CSPs with digraph CSPs, we now present it as stand-alone statement.
\begin{theorem}\label{thm:logspace}
  Every fixed finite template CSP is logspace equivalent to the CSP
  over some finite digraph.
\end{theorem}

\section{Preserving cores}
In what follows, let $\mathbb A$ be a fixed finite relational
structure. Without loss of generality we may assume that $\mathbb
A=(A;R)$, where $R$ is a $k$-ary relation
(see Example~\ref{example:1}).

\begin{lemma}\label{lem:end} 
  The endomorphisms of $\mathbb A$ and $\mathcal{D}(\mathbb A)$ are in
  one-to-one correspondence.
\end{lemma}
\begin{proof} 
  We first show that every endomorphism $\varphi$ of $\mathbb A$ can
  be extended to an endomorphism $\overline{\varphi}$ of $\mathcal
  D(\mathbb A)$.  Let $\overline{\varphi}(a)=\varphi(a)$ for $a\in A$,
  and let $\overline{\varphi}(\mathbf r)=\varphi^{(k)}(\mathbf r)$ for
  $\mathbf r\in R$. Note that $\varphi^{(k)}(\mathbf r)\in R$ since
  $\varphi$ is an endomorphism of $\mathbb A$.

  Let $c\in\mathcal D(\mathbb A)\setminus(A\cup R)$ and let
  $e=(a,\mathbf r)$ be such that $c\in\mathbb P_e$. Define
  $e'=(\varphi(a),\varphi^{(k)}(\mathbf r))$. If $\mathbb P_{e,l}$ is
  a single edge for some $l\in[k]$, then $r_l=a$ and
  $\varphi(r_l)=\varphi(a)$, and therefore $\mathbb P_{e',l}$ is a
  single edge. Thus there exists a (unique) homomorphism $\mathbb
  P_e\to\mathbb P_{e'}$. Define $\overline{\varphi}(c)$ to be the
  image of $c$ under this homomorphism, completing the definition of
  $\overline{\varphi}$.

  We now show that every endomorphism $\Phi$ of $\mathcal{D}(\mathbb
  A)$ is of the form $\overline{\varphi}$, for some endomorphism
  $\varphi$ of $\mathbb A$.  Let $\Phi$ be an endomorphism of
  $\mathcal D(\mathbb A)$. Let $\varphi$ be the restriction of $\Phi$
  to $A$. By Lemma \ref{lem:ppdef_polymorphisms} and Lemma
  \ref{lem:forwardreduction}, $\varphi$ is an endomorphism of $\mathbb
  A$.  For every $e=(a,\mathbf r)$, the endomorphism $\Phi$ maps
  $\mathbb P_e$ onto $\mathbb P_{(\varphi(a),\Phi(\mathbf r))}$. If we
  set $a=r_l$, then $\mathbb P_{e,l}$ is a single edge. In this case
  it follows that $\mathbb P_{(\varphi(a),\Phi(\mathbf r)),l}$ is also
  a single edge. Thus, by the construction of $\mathcal D(\mathbb A)$
  the $l^\text{th}$ coordinate of $\Phi(\mathbf r)$ is
  $\varphi(a)=\varphi(r_l)$. This proves that the restriction of
  $\Phi$ to $R$ is $\varphi^{(k)}$ and therefore
  $\Phi=\overline{\varphi}$.\hfill$\qed$
\end{proof}

The following corollary is Theorem~\ref{thm:1} (iii).

\begin{corollary}\label{cor:core} 
  $\mathbb{A}$ is a core if and only if $\mathcal D(\mathbb A)$ is a
  core.
\end{corollary}

\begin{proof}
  To prove the corollary we need to show that an endomorphism
  $\varphi$ of $\mathbb A$ is surjective if and only if
  $\overline{\varphi}$ (from Lemma~\ref{lem:end}) is
  surjective. Clearly, if $\overline{\varphi}$ is surjective then so
  is $\varphi$.

  Assume $\varphi$ is surjective (and therefore an automorphism of
  $\mathbb A$). It follows that $\varphi^{(k)}$ is surjective on $R$ and therefore
  $\overline{\varphi}$ is a bijection when restricted to the set
  $A\cup R$.  Let $a\in A$ and $\mathbf r \in R$. By definition we
  know that $\overline{\varphi}$ maps $\mathbb{P}_{(a,\mathbf r)}$
  homomorphically onto $\mathbb{P}_{(\varphi(a),\varphi^{(k)}(\mathbf
    r))}$. Since $\varphi$ has an inverse $\varphi^{-1}$, it follows
  that $\overline{\varphi^{-1}}$ maps
  $\mathbb{P}_{(\varphi(a),\varphi^{(k)}(\mathbf r))}$ homomorphically
  onto $\mathbb{P}_{(a,\mathbf r)}$.  Thus $\mathbb{P}_{(a,\mathbf
    r)}$ and $\mathbb{P}_{(\varphi(a),\varphi^{(k)}(\mathbf r))}$ are
  isomorphic, completing the proof.\hfill$\qed$
\end{proof}

To complete the proof of Theorem~\ref{thm:1}, it remains to show that
our reduction preserves WNUs.

\section{The reduction preserves WNUs}
For $m>0$ let $\Delta_m$ denote the connected component of the digraph
$\mathcal D(\mathbb A)^m$ containing the \emph{diagonal} (i.e., the
set $\{(c,c,\dots,c)\mid c\in\mathcal D(\mathbb A)\}$). 
\begin{lemma}\label{obs:2}
  The elements of the diagonal are all connected in
  $\mathcal{D}(\mathbb A)^m$. Furthermore, $A^m\subseteq\Delta_m$ and
  $R^m\subseteq\Delta_m$.
\end{lemma}
\begin{proof}
  The first statement follows from the fact that $\mathcal{D}(\mathbb
  A)$ is connected. To see that $A^m\subseteq\Delta_m$ and
  $R^m\subseteq\Delta_m$, fix $a\in A$, and so by definition,
  $(a,a,\dots,a)\in \Delta_m$. Let $(\mathbf r^1,\dots,\mathbf r^m)\in
  R^m$ and for every $i\in[m]$ let $\varphi_i:\mathbb Q\to\mathbb
  P_{(a,\mathbf r^i)}$. The homomorphism defined by $x\mapsto
  (\varphi_1(x),\dots,\varphi_m(x))$ witnesses
  $(a,\dots,a)\stackrel{\mathbb Q}{\longrightarrow}(\mathbf
  r^1,\dots,\mathbf r^m)$ in $\mathcal D(\mathbb A)^m$. Hence
  $R^m\subseteq\Delta_m$. A similar argument gives
  $A^m\subseteq\Delta_m$.~$\qed$
\end{proof}

The following lemma shows that there is only one non-trivial connected
component of $\mathcal{D}(\mathbb{A})^m$ that contains tuples (whose
entries are) on the same level in $\mathcal{D}(\mathbb{A})$; namely
$\Delta_m$. 

\begin{lemma}\label{lem:diagonalcomponent}
  Let $m>0$ and let $\Gamma$ be a connected component of $\mathcal
  D(\mathbb A)^m$ containing an element $\mathbf c$ such that
  $\lvl(c_1)=\dots=\lvl(c_m)$. Then every element
  $\mathbf d \in \Gamma$ is of the form $\lvl(d_1)=\dots = 
  \lvl (d_m)$ and the following hold.
\begin{enumerate}[(i)]
\item If $\mathbf c\to\mathbf d$ is an edge in $\Gamma$ such that
  $\mathbf c\notin A^m$ and $\mathbf d\notin R^m$, then there exist
  $e_1,\dots,e_m\in A\times R$ and $l\in[k]$ such that $\mathbf
  c,\mathbf d\in\prod_{i=1}^m\mathbb P_{e_i,l}$.
\item Either $\Gamma=\Delta_m$ or $\Gamma$ is one-element.
\end{enumerate}
\end{lemma}
\begin{proof}
  First observe that if an element $\mathbf d$ is connected in
  $\mathcal{D}(\mathbb A)^m$ to an element $\mathbf c$ with
  $\lvl (c_1)= \dots = \lvl (c_m)$, then there is an
  oriented path $\mathbb{Q}'$ such that $\mathbf c \stackrel{\mathbb
    Q'}\to \mathbf d$ from which it follows that $\lvl (d_1)=
  \dots = \lvl (d_m)$.

    To prove (i), let $\mathbf c\to\mathbf d$ be an edge in $\Gamma$
  such that $\mathbf c\notin A^m$ and $\mathbf d\notin R^m$. For
  $i=1,\dots, m$ let $e_i$ be such that $c_i\in\mathbb P_{e_i}$ and
  let $l=\lvl (c_1)$. The claim now follows immediately from the
  construction of $\mathcal D(\mathbb A)$.
  
  It remains to prove (ii).  If $|\Gamma|>1$, then there is an edge
  $\mathbf c\to\mathbf d$ in $\Gamma$. If $\mathbf c\in A^m$ or
  $\mathbf d\in R^m$, then the result follows from
  Lemma~\ref{obs:2}. Otherwise, from~(i), there exists $l\in[k]$ and
  $e_i=(a_i,\mathbf r^i)$ such that $\mathbf c,\mathbf
  d\in\prod_{i=1}^m\mathbb P_{e_i,l}$.

  For every $i\in[m]$ we can walk from $c_i$ to $\iota\mathbb
  P_{e_i,l}$ following the path
  \mbox{$\bullet\boldsymbol\rightarrow\bullet\boldsymbol\leftarrow\bullet$};
  and so $\mathbf c$ and $(\iota\mathbb P_{e_1,l},\dots,\iota\mathbb
  P_{e_m,l})$ are connected. For every $i\in[m]$ there exists a
  homomorphism $\varphi_i:\mathbb Q\to\mathbb P_{e_i}$ such that
  $\varphi_i(\iota\mathbb Q)=a_i$ and $\varphi_i(\iota\mathbb
  Q_l)=\iota\mathbb P_{e_i,l}$. The homomorphism $\mathbb Q\to
  \mathcal D(\mathbb A)^m$ defined by $x\mapsto
  (\varphi_1(x),\dots,\varphi_m(x))$ shows that $(a_1,\dots,a_m)$ and
  $(\iota\mathbb P_{e_1,l},\dots,\iota\mathbb P_{e_m,l})$ are
  connected. By transitivity, $(a_1,\dots,a_m)$ is connected to
  $\mathbf c$ and therefore $(a_1,\dots,a_m)\in \Gamma$. Using (i) we
  obtain $\Gamma=\Delta_m$.\hfill$\qed$
\end{proof}

We are now ready to prove the main ingredient of Theorem~\ref{thm:1}
(ii).  The proof of Lemma~\ref{lem:preserves_WNUs} is similar in
essence to the proof of Lemma~\ref{lem:end}, although more
complicated.

\begin{lemma}\label{lem:preserves_WNUs} 
  If $\mathbb A$ has an $m$-ary WNU polymorphism, then $\mathcal
  D(\mathbb A)$ has an $m$-ary WNU polymorphism.
\end{lemma}
\begin{proof}
  Let $\omega$ be an $m$-ary WNU polymorphism of $\mathbb{A}$. We
  construct an $m$-ary operation $\overline{\omega}$ on $\mathcal
  D(\mathbb A)$. We split the definition into several cases and
  subcases.  Let $\mathbf c\in \mathcal D(\mathbb A)^m$.

\medskip

\noindent\textbf{Case 1.} $|\{\lvl (c_1),\dots,\lvl (c_m)\}|>1$.

\smallskip

\noindent\underline{1a}\quad If there exists $i\in[m]$ such that
$|\{\lvl (c_j)\mid j\neq i\}|=1$, we define $\overline{\omega}(\mathbf
c)=c_i$.

\smallskip

\noindent\underline{1b}\quad Else let $\overline{\omega}(\mathbf c)=c_1$.

\smallskip

\noindent\textbf{Case 2.} $\lvl (c_1)=\dots=\lvl (c_m)$,
but $\mathbf c\notin\Delta_m$.

\smallskip

\noindent\underline{2a}\quad If there exists $i\in[m]$ such that $|\{c_j\mid j\neq i\}|=1$,
  we define $\overline{\omega}(\mathbf c)=c_i$.

\smallskip

\noindent\underline{2b}\quad Else let $\overline{\omega}(\mathbf c)=c_1$.

\smallskip

\noindent\textbf{Case 3.} $\mathbf c\in\Delta_m$.

\smallskip

\noindent\underline{3a}\quad If $\{c_1,\dots,c_m\}\subseteq A$, we
define
  $\overline{\omega}(\mathbf c)=\omega(\mathbf c)$.

\smallskip

\noindent\underline{3b}\quad If $\{c_1,\dots,c_m\}\subseteq R$, we
define
  $\overline{\omega}(\mathbf c)=\omega^{(k)}(\mathbf c)$.

\smallskip

\noindent\underline{3c}\quad Else, there exists $\mathbf
d\in\Delta_m\setminus(A^m\cup R^m)$ such that $\mathbf
c\rightarrow\mathbf d$ or $\mathbf d\rightarrow\mathbf c$ in $\mathcal
D(\mathbb A)^m$.  By Lemma \ref{lem:diagonalcomponent} (ii), there
exist $l\in[k]$ and $e_i=(a_i,\mathbf r^i)$ such that $\mathbf
c\in\prod_{i=1}^m\mathbb P_{e_i,l}$. Let $e=(a,\mathbf r)$, where
$a=\omega(a_1,\dots,a_m)$ and $\mathbf r=\omega^{(k)}(\mathbf
r^1,\dots,\mathbf r^m)$. We set $\overline{\omega}(\mathbf
c)=\Phi(\mathbf c)$, where $\Phi:\prod_{i=1}^m\mathbb
P_{e_i,l}\to\mathbb P_{e,l}$ is defined as follows.
\begin{enumerate}
\item If $\mathbb P_{e,l}$ is a single edge, then we set
  \[
    \Phi(\mathbf u)=\begin{cases}\iota\mathbb P_{e,l} & \text{ if }
      \lvl (u_1)=\dots=\lvl (u_m)=\lvl (\iota\mathbb P_{e,l})\\
      \tau\mathbb P_{e,l} & \text{ otherwise.}
    \end{cases}
    \]
  \item If $\mathbb P_{e,l}$ is a zigzag, then let
    $I=\{i\in[m]\mid\mathbb P_{e_i,l}\ \text{is a zigzag}\}$. For
    every $i\in I$ let $\phi_i:\mathbb P_{e_i,l}\to\mathbb P_{e,l}$ be
    the unique isomorphism. We define $\Phi(\mathbf u)$ to be the
    vertex from $\{\phi_i(u_i):i\in I\}$ with minimal distance from
    $\iota\mathbb P_{e,l}$.
\end{enumerate}

\medskip

Let us first comment on correctness of the definition. In subcase
\underline{3b}\,, $\omega^{(k)}(\mathbf c)\in R$ follows from the fact
that $\omega$ preserves $R$. In subcase \underline{3c}\,, if $\mathbb
P_{e,l}$ is a zigzag, then $I\neq\emptyset$. Indeed, if all the
$\mathbb P_{e_i,l}$'s were single edges, then
$r_l=\omega(r^1_l,\dots,r^m_l)=\omega(a_1,\dots,a_m)=a$ and so
$\mathbb P_{e,l}$ would also be a single edge. The $e_i$'s are
uniquely determined by $\mathbf c$, and the choice of $l$ is unique as
well, with one exception: if $l<k$ and $c_i=\tau\mathbb
P_{e_i,l}=\iota\mathbb P_{e_i,l+1}$ for every $i\in[m]$, then we have
$\mathbf d\to\mathbf c\to\mathbf{d'}$ for some $\mathbf
d,\mathbf{d'}\in\mathcal D(\mathbb A)^m$ and we can choose $l+1$
instead of $l$. However, it is not hard to see that the value assigned
to $\overline{\omega}(\mathbf c)$ is the same in both cases, namely it
is the vertex $\tau\mathbb P_{e,l}=\iota\mathbb P_{e,l+1}$ (see
property (b) below).

We need the following properties of the mapping $\Phi$ defined in
\underline{3c}\,.
\begin{enumerate}[(a)]
\item $\Phi$ is a homomorphism.
\item $\Phi(\iota\mathbb P_{e_1,l},\dots,\iota\mathbb
  P_{e_m,l})=\iota\mathbb P_{e,l}$ and $\Phi(\tau\mathbb
  P_{e_1,l},\dots,\tau\mathbb P_{e_m,l})=\tau\mathbb P_{e,l}$.
\item $\Phi$ does not depend on the ordering of the tuple $(e_1,\dots,e_m)$.
\item If $e_1=\dots=e_m=e$, then $\Phi:\mathbb P_{e,l}^m\to\mathbb
  P_{e,l}$ is idempotent, i.e., $\Phi(u,\dots,u)=u$ for all
  $u\in\mathbb P_{e,l}$.
\end{enumerate}
All of the above properties follow easily from the definition of
$\Phi$. We leave the verification to the reader.  It remains to prove
that $\overline{\omega}$ is a WNU polymorphism of $\mathcal D(\mathbb
A)$.
\begin{claim}
  $\overline{\omega}$ is idempotent and satisfies the WNU equations.
\end{claim}
Let $c,d\in\mathcal D(\mathbb A)$. Note that all of the tuples
$(c,\dots,c,d),(c,\dots,d,c),\dots,(d,c,\dots,c)$ fall into the same
subcase; the possibilities are \underline{3a}\,, \underline{3b}\,,
\underline{3c} and if $c\neq d$, then also \underline{1a} or
\underline{2a}\,. In subcases \underline{1a} and \underline{2a} the
definition does not depend on the ordering of the input tuple at all;
therefore the WNU equations hold (and since $c\neq d$, idempotency
does not apply). 

In case 3 we use the fact that $\omega$ and $\omega^{(k)}$ are
idempotent and satisfy the WNU equations. In \underline{3a} and
\underline{3b} the claim follows immediately. In \underline{3c} note
that $e$ is the same for all of the tuples
$(c,\dots,c,d),\dots,(d,c,\dots,c)$, and if $c=d$, then
$e_1=\dots=e_m=e$. The WNU equations follow from property (c) of the
mapping $\Phi$, and idempotency follows from (d).

\begin{claim}
$\overline{\omega}$ is a polymorphism of $\mathcal D(\mathbb A)$.
\end{claim}
Let $\mathbf c,\mathbf d\in\mathcal D(\mathbb A)^m$ be such that
$\mathbf c\rightarrow\mathbf d$ is an edge in $\mathcal D(\mathbb
A)^m$, that is, $c_i\to d_i$ for all $i\in[m]$. Both the tuples
$\mathbf c$ and $\mathbf d$ fall into the same case and, by Lemma
\ref{lem:diagonalcomponent}\,(iii), it cannot be case 2. If it is case
1, then they also fall into the same subcase and it is easily seen
that $\overline{\omega}(\mathbf c)\to\overline{\omega}(\mathbf d)$.

If $\mathbf c$ falls into subcase \underline{3a}\,, then $\mathbf d$
falls into \underline{3c}\,. Let $e_i$ be such that $d_i\in\mathbb
P_{e_i}$. As $c_i\to d_i$, it follows that $e_i=(c_i,\mathbf r^i)$ for
some $\mathbf r^i\in R$ and $d_i=\iota\mathbb P_{e_i,1}$. Let us
define $c=\omega(\mathbf c)$, $\mathbf r=\omega^{(k)}(\mathbf
r^1,\dots,\mathbf r^m)$ and $e=(c,\mathbf r)$. Now
$\overline{\omega}(\mathbf d)$ is the result of the mapping $\Phi$
applied to the tuple of initial vertices of the $\mathbb P_{e_i,1}$'s,
which (by property (b)) is the initial vertex of $\mathbb P_{e,1}$. So
$\overline{\omega}(\mathbf c)=c\to\iota\mathbb
P_{e,1}=\overline{\omega}(\mathbf d)$. The argument is similar if
$\mathbf d$ falls into \underline{3b}\,.

It remains to verify that $\overline{\omega}(\mathbf
c)\to\overline{\omega}(\mathbf d)$ if both $\mathbf c$ and $\mathbf d$
fall into subcase \underline{3c}\,. Both $\overline{\omega}(\mathbf
c)$ and $\overline{\omega}(\mathbf d)$ are defined using the mapping
$\Phi:\prod_{i=1}^m\mathbb P_{e_i,l}\to\mathbb P_{e,l}$. Since $\Phi$
is a homomorphism, we have $\overline{\omega}(\mathbf c)=\Phi(\mathbf
c)\to\Phi(\mathbf d)=\overline{\omega}(\mathbf d)$, concluding the
proof.\hfill$\qed$
\end{proof}

\section{Discussion}
The algebraic dichotomy conjecture proposes a polymorphism
characterisation of tractability for core CSPs.  A number of other
algorithmic properties are also either proved or conjectured to
correspond to the existence of polymorphisms with special equational
properties.  For instance, solvability by the few subpowers algorithm
(a generalization of Gaussian elimination) as described in Idziak et
al.~\cite{IMMVW} has a polymorphism characterisation \cite{BIMMVW}, as
well as problems of bounded width (see Theorem~\ref{thm:many_WNUs}).
The final statement in Theorem \ref{thm:1} already shows that
$\mathbb{A}$ has bounded width if and only if
$\mathcal{D}(\mathbb{A})$ has bounded width. Kazda~\cite{kaz} showed
that every digraph with a Maltsev polymorphism must have a majority
polymorphism, which is not the case for finite relational structures
in general. In a later version of the present article we will show
that Theorem \ref{thm:1}(ii) extends to include almost all commonly
encountered polymorphism properties aside from Maltsev.  For instance
the CSP over $\mathbb{A}$ is solvable by the few subpowers algorithm
if and only if the same is true for $\mathcal D(\mathbb A)$. Among
other conditions preserved under our reduction to digraphs is that of
arc consistency (or width 1 problems~\cite{width1}) and problems of
bounded strict width~\cite{fedvar}.  The following example is a
powerful consequence of the result.
\begin{example}\label{eg:fsp}
Let $\mathbb{A}$ be the structure on $\{0,1\}$ with a single $4$-ary relation 
\[
\{(0,0,0,1),(0,1,1,1),(1,0,1,1),(1,1,0,1)\}.
\]  
Clearly $\mathbb A$ is a core. The polymorphisms of $\mathbb{A}$ can
be shown to be the idempotent term functions of the two element group,
and from this it follows that $\CSP(\mathbb{A})$ is solvable by the
few subpowers algorithm of \cite{IMMVW}, but is not bounded width.
Then the CSP over the digraph $\mathcal{D}(\mathbb{A})$ is also
solvable by few subpowers but is not bounded width (that is, is not
solvable by local consistency check).
\end{example}
Prior to the announcement of this example it had been temporarily
conjectured by some researchers that solvability by the few subpowers
algorithm implied solvability by local consistency check in the case
of digraphs (this was the opening conjecture in Mar{\' o}ti's keynote
presentation at the Second International Conference on Order, Algebra
and Logics in Krakow 2011 for example).  With 78 vertices and 80
edges, Example \ref{eg:fsp} also serves as a simpler alternative to
the 368-vertex, 432-edge digraph whose CSP was shown by Atserias
in~\cite[\S4.2]{ats} to be tractable but not solvable by local
consistency check.

There are also conjectured polymorphism classifications of the
property of solvability within nondeterministic logspace and within
logspace; see Larose and Tesson \cite{lartes}.  The required
polymorphism conditions are among those we can show are preserved
under the transition from $\mathbb{A}$ to $\mathcal{D}(\mathbb{A})$.
It then follows that these conjectures are true provided they can be
established in the restricted case of CSPs over digraphs.

\noindent\textbf{Open problems.}
We conclude our paper with some further research directions. \\
It is possible to show that the logspace reduction in Lemma
\ref{lem:reversereduction} \emph{cannot} be replaced by first order
reductions.  Is there a different construction that translates general
CSPs to digraph CSPs with first order reductions in both directions?

Feder and Vardi~\cite{fedvar} and Atserias~\cite{ats} provide
polynomial time reductions of CSPs to digraph CSPs.  We vigorously
conjecture that their reductions preserve the properties of possessing
a WNU polymorphism (and of being cores; but this is routinely
verified).  Do these or other constructions preserve the precise arity
of WNU polymorphisms?  What other polymorphism properties are
preserved? Do they preserve the bounded width property? Can they
preserve \emph{conservative} polymorphisms (the polymorphisms related
to list homomorphism problems)?  The third and fourth authors with
Kowalski \cite{JKN13} have recently shown that a minor variation of
the $\mathcal{D}$ construction in the present article preserves
$k$-ary WNU polymorphisms (and can serve as an alternative to
$\mathcal{D}$ in Theorem \ref{thm:1}), but always fails to preserve
many other polymorphism properties (such as those witnessing strict
width).

\section{Acknowledgements}
\noindent The authors would like to thank Libor Barto, Marcin Kozik,
Mikl\'os Mar\'oti and Barnaby Martin for their thoughtful comments and
discussions.

\bibliography{bibCP2013}

\appendix

\section{The logspace reduction of $\CSP(\mathcal{D}(\A))$ to $\CSP(\A)$.}
In this appendix, we give the proof of Lemma
\ref{lem:reversereduction}, by showing that $\CSP(\mathcal{D}(\A))$
reduces in logspace to $\CSP(\A)$.  A sketch of a \emph{polymomial
  time} reduction is given in the proof of \cite[Theorem 13]{fedvar};
technically, that argument is for the special case where $\mathbb{A}$
is itself already a digraph, but the arguments can be broadened to
cover our case.  To perform this process in logspace is rather
technical, with many of the difficulties lying in details that are
omitted in the polymomial time description in the proof of
\cite[Theorem 13]{fedvar}.  We wish to thank Barnaby Martin for
encouraging us to pursue Lemma \ref{lem:reversereduction}.

\subsection{Outline of the algorithm.}
We first assume that $\CSP(\mathbb{A})$ is itself not trivial (that
is, that there is at least one no instance and one yes instance): this
uninteresting restriction is necessary because
$\CSP(\mathcal{D}(\mathcal{A}))$ will have no instances always.  Now
let $\G$ be an instance of $\CSP(\mathcal{D}(\A))$.  Also, let $n$
denote the height of $\mathcal{D}(\A)$ and $k$ the arity of the single
fundamental relation $R$ of $\A$: so, $n=k+2$. Recall that the
vertices of $\mathcal{D}(\A)$ include those of $A$ as well as the
elements in $R$.  The rough outline of the algorithm is as follows.

\begin{enumerate}
\item[] (Stage 1.) Some initial analysis of $\G$ is performed to decide if it is broadly of the right kind of digraph to be a possible yes instance.  If not, some fixed no instance of $\CSP(\A)$ is output.  
\item[] (Stage 2.) It is convenient to remove any components of $\G$
  that are too small.  These are considered directly, and in logspace
  we determine whether or not they are YES or NO instances of
  $\CSP(\mathcal{D}(\A))$.  If all are YES we ignore them.  If one
  returns NO we reject the entire instance and return some fixed NO
  instance of $\CSP(\A)$.
\item[] (Stage 3.) Now it may be assumed that $\G$ is roughly similar
  to a digraph of the form $\mathcal{D}(\B)$, but where some vertices
  at level $0$ have been lost, and other vertices at this level and at
  level $n$ have been split into numerous copies, with each possibly
  containing different parts of the information in the connecting
  edges of $\mathcal{D}(\B)$.  Essentially, the required $\B$ is a
  kind of quotient of an object definable from $\G$, though some extra
  vertices must be added (this is similar to the addition of vertices
  to account for existentially quantified variables in a primitive
  positive definition of a relation: only new vertices are added, and
  they are essentially unconstrained beyond the specific purpose for
  which they are added).  To construct $\B$ in logspace, we work in
  two steps: we describe a logspace construction of some intermediate
  information.  Then we describe a logspace reduction from strings of
  suitable information of this kind to $\B$.  The overall process is
  logspace because a composition of two logspace reductions is a
  logspace reduction.
\item[] (Stage 3A.)  From $\G$ we output a list of ``generalised
  hyperedges'' consisting of $k$-tuples of sets of vertices, with some
  labelling to record how these were created. These generalised
  hyperedges arise from $\G$ but each position in the generalised
  hyperedge consists of a set of vertices of $\G$, possibly including
  some new vertices, instead of individual
  vertices. 
\item[] (Step 3B.) The actual structure $\B$ is constructed from the
  generalised hyperedges in the previous step.  The input consists of
  generalised hyperedges.  To create $\B$ numerous undirected graph
  reachability checks are performed.  The final ``vertices'' of $\B$
  are in fact sets of vertices of $\G$, so that the generalised
  hyperedges become actual hyperedges in the conventional sense
  ($k$-tuples of ``vertices'', now consisting of sets of vertices of
  $\G$).  This may be reduced to an adjacency matrix description as a
  separate logspace process, but that is routine.
\end{enumerate}

Stage 1 is described in Subsection \ref{subsec:1}, while Stage 2 is
described in Subsection \ref{subsec:2}.  The most involved part of the
algorithm is stage 3A.  In Subsection \ref{subsec:3} we give some
preliminary discussion on how the process is to proceed: an
elaboration on the item listed in the present subsection.  In
particular a number of definitions are introduced to aid the
description of Stage 3A.  The actual algorithm is detailed in
Subsection \ref{subsec:3A}.  Step 3B is described in Subsection
\ref{subsec:3B}.  After a brief discussion of why the algorithm is a
valid reduction from $\CSP(\mathcal{D}(\A))$ to $\CSP(\A)$, we present
an example of Stages 3A and 3B in action.  This example may be a
useful reference while reading Subsection \ref{subsec:3A} and
\ref{subsec:3B}.

Before we begin describing the algorithm we recall some basic logspace
process that we will use frequently.
\subsection{Subroutines}
The algorithm we describe makes numerous calls on other logspace
computable processes.  Our algorithm may be thought of as running on
an oracle machine, with several query tapes.  Each query tape verifies
membership in some logspace solvable problem.  It is well known that
$\texttt{L}^\texttt{L}=\texttt{L}$, and this enables all of the query
tapes to be eliminated within logspace.  For the sake of clarity, we
briefly recall some basic information on logspace on an oracle
machine.  A oracle program with logspace query language $U$, has
access to an input tape, a working tape (or tapes) and an output tape.
Unlimited reading may be done from the input, but no writing.
Unlimited writing may be done to the output tape, but no reading.
Unlimited writing may be done to the query tape, but no reading.  Once
the query state is reached however, the current word written to the
query tape is tested for membership in the language $U$ (at the cost
of one step of computation), and a (correct) answer of either yes or
no is received by the program, and the query tape is immediately
erased.  The space used is measured \emph{only from the working tape},
where both reading and writing is allowed.  If such a program runs in
logspace, then it can be emulated by an actual logspace program (with
no oracle), so that $\texttt{L}^\texttt{L}=\texttt{L}$.  The argument
is essentially the fact that a composition of logspace reductions is a
logspace reduction: each query to the oracle (of a string $w$ for
instance) during the computation is treated as a fresh instance of a
reduction to the membership problem of $U$, which is then composed
with the logspace algorithm for $U$ (which is, as usual, done without
ever writing any more than around one symbol of $w$ at a time---plus a
short counter---which is why space used on the oracle tape does not
matter in the oracle formulation of logspace, and why we may assume
that the query tape may be erased after completion of the query).

In this subsection we describe the basic checks that are employed
during or algorithm.

{\bf Undirected reachability.}  Given an undirected graph and two
vertices $u,v$, there is a logspace algorithm to determine if $u$ is
reachable from $v$ (Reingold \cite{rei05,rei08}).  In the case of an
undirected graph we may use this to determine if two vertices are
connected by some oriented path (simply treat the graph as an
undirected graph, and use undirected reachability).  This means, for
example, that we can construct, in logspace, the \emph{smallest
  equivalence relation containing some input binary relation}.
Starting with $i=0$ and $j=i+1$, consider the $i^{\rm th}$ vertex.
Check to see if any earlier vertices are weakly connected using
queries to the undirected reachability test.  If so, then increment
$i$ and $j$.  If not, then initiate the output of the block of the
equivalence relation containing vertex $i$.  We may output $i$
(possibly ``$\{i$'' if we wish to output the blocks of the equivalence
relation in set notation), and then query whether the $j^{\rm th}$
vertex is undirected reachable from the $i^{\rm th}$.  If so, then
output the $j^{\rm th}$ vertex, otherwise ignore the $j^{\rm th}$
vertex.  Then increment $j$ and repeat these checks until all vertices
have been checked (and close the brackets~$\}$).  Now increment $i$
and repeat.

A second process we frequently perform is undirected reachability
checks in graphs whose edges are not precisely those of the current
input graph.  A typical instance might be where we have some fixed
vertex $u$ in consideration, and we wish to test if some vertex $v$
can be reached from $u$ by an oriented path avoiding all vertices
failing some property $\mathcal{Q}$, where $\mathcal{Q}$ is a logspace
testable property.  This is basically undirected graph reachability,
except as well as ignoring the edge direction, we must also ignore any
vertex failing property $\mathcal{Q}$.  This is again logspace,
because it can be performed in logspace on an oracle machine running
an algorithm for undirected graph reachability and whose query tape
tests property $\mathcal{Q}$.

{\bf Component checking.} Undirected graph reachability is also
fundamental to checking properties of induced subgraphs.  In a typical
situation we have some induced subgraph $C$ of $\G$ (containing some
vertex $u$, say) and we want to test if it satisfies some property
$\mathcal{P}$.  Membership of vertices in $C$ is itself determined by
some property $\mathcal{Q}$, testable in logspace.  It is convenient
to assume that the query tape for $\mathcal{P}$ expects inputs that
consist of a list of directed edges (if adjacency matrix is preferred,
then this involves one further nested logspace process, but the
argument is routine).  We may construct a list of the directed edges
in the component $C$ on a logspace machine with a query tape for
$\mathcal{P}$, for undirected graph reachability and for
$\mathcal{Q}$.  We write $C$ to the query tape for $\mathcal{P}$ as
follows.  Systematically enumerate pairs of vertices $v_1,v_2$ of $\G$
(re-using some fixed portion of work tape for each pair), in each case
testing for undirected reachability of both $v_1$ and $v_2$ from $u$,
and also for satisfaction of property $\mathcal{Q}$.  If both are
reachable, and if $(v_1,v_2)$ is an edge of $\G$ then we output the
edge $(v_1,v_2)$ to the query tape for $\mathcal{P}$.  After the last
pair has been considered, we may finally query~$\mathcal{P}$.

{\bf Testing for path satisfaction.} The basic properties we wish to
test of components are usually satisfiability within some fixed finite
family of directed paths.  We consider the paths $\Q_S$, where $S$ is
some subset of $[k]=\{1,\dots,k\}$: recall that these have zigzags in
a position $i$ when $i\notin S$.  We first argue that it suffices to
show that such queries are logspace in the the case of singleton $S$.
Consider $\Q_{[k]\backslash\{i\}}$, endowed with not just the edge
relation, but also all singleton unary constraints.  Because
$\Q_{[k]\backslash\{i\}}$ is a core digraph, it is well known that the
CSP over $\Q_{[k]\backslash\{i\}}$ is logspace equivalent to the CSP
over $\mathbb{Q}_{[k]\backslash\{i\}}$ with singleton unary relations
added.

For a typical input digraph $\H$, there will be a smallest set
$S\subseteq\{1,\dots,k\}$ such that $H$ can be satisfied in
$\Q_S$. This can be found by testing for satisfaction in
$\mathbb{Q}_{[k]\backslash\{i\}}$, where $i$ ranges from $1$ to $k$.
The set $S$ will consist of those $i$ for which
$\H\notin\CSP(\Q_{[k]\backslash\{i\}})$.  A technicality: if the
height of $\H$ is less than $n$, then it may be satisfiable within
some $\Q_{[k]\backslash\{i\}}$ at different levels.  In order to
combine the positive $i$ into the set $S$, it is necessary that these
interpretations be done at the same level.  This can be done by fixing
some vertex of $\H$, and using a unary singleton constraint to force
it to be interpreted at some fixed level; testing all
$\Q_{[k]\backslash\{i\}}$.  In almost every case we encounter however,
we will know in advance that certain vertices of $\H$ \emph{must} be
satisfied at either the unique vertex adjacent to $\iota\Q_S$, or to
$\tau\Q_S$.  Thus we only look for such interpretations.
\begin{lemma}\label{lem:pathassign}
\begin{enumerate}
\item $\CSP(\Q_{[k]\backslash\{i\}})$ is solvable in logspace for any
  $i\in \{1,\dots,k\}$, even when singleton unary relations are added.
\item For any $S\subseteq\{1,\dots,k\}$ the problem $\CSP(\Q_S)$ is solvable in logspace, including when particular vertices are constrained to be satisfied at particular points.
\item If $\Q_{S_1}$, $\Q_{S_2}$,\dots, $\Q_{S_\ell}$ is a family of
  connecting paths in $\mathcal{D}(\A)$, then the CSP over the digraph
  formed by amalgamating the $\Q_{S_i}$ at either all the initial
  points, or at all the terminal points is logspace solvable.
\end{enumerate}
\end{lemma}
\begin{proof}
  The first item is a special case of the second, however the proof of
  the second case will use the first case.

  (1). This is crucial and not trivial ($\Q_{[k]\backslash\{i\}}$ does
  not have a Maltsev polymorphism, so it is not covered under the most
  widely used polymorphism techniques for establishing logspace
  solvability). Note that as $\Q_{[k]\backslash\{i\}}$ is a core, we
  have $\CSP(\Q_{[k]\backslash\{i\}})$ logspace equivalent to the CSP
  over $\Q_{[k]\backslash\{i\}}$ with all unary singletons added.

  Given an input digraph $\H$, we first test if $\H$ is satisfiable in
  the directed path $\Pp_n$ of height $n$ (checking that $\H$ is
  balanced, and of sufficiently small height).  Reject if NO.
  Otherwise we may assume that $\H$ is a single component (because it
  suffices to satisfy each component, and component checking has been
  described as logspace in earlier discussion).

  Fix any vertex $u$ of $\H$.  It is trivial that $\H$ is satisfiable
  in $\Q_{[k]\backslash\{i\}}$ if and only it is satisfied with $u$
  interpreted at \emph{some height} in $\Q_{[k]\backslash\{i\}}$.  We
  will fix a possible height for $u$, and successively try new heights
  if this fails until all are exhausted or $\H$ is found to be a Yes
  instance.  Once we have fixed a target height $j$ for $u$, this
  determines the height of all other vertices in $\H$ (this is the
  only place where $\H$ being one component is used). For example, to
  determine the height of another vertex $v$, we may successively test
  for satisfaction of $\H$ in $\Q_{[k]}$, with $u$ constrained at
  fixed height $j$, and $v$ constrained at height $0$, then $1$, then
  $2$ and so on, until a YES is returned.  There are only $n+1$
  heights, and each is just a query to a logspace oracle.  Because
  $\H$ is balanced and connected, there is at most one height at which
  $v$ will produce a YES.  If none is produced, then $u$ cannot
  interpret in $\Q_{[k]}$ at height $j$, so that $j$ itself must be
  incremented.

  In this way, we may assume within logspace that we have access to
  the heights of all vertices during the following argument.

  Now, in any satisfying interpretation of $\H$ in
  $\Q_{[k]\backslash\{i\}}$, any vertices of the same height $j\notin
  \{i,i+1\}$ are identified.  We need only ensure there is no path of
  vertices of height $i-1,i,i+1,i+2$.  So it suffices to enumerate all
  $4$-tuples of vertices $u_1,\dots,u_4$, check if $u_1\rightarrow u_2
  \rightarrow u_3\rightarrow u_4$, and if so, check that the height of
  $u_1$ is not $i-1$.  If it is, then reject.  Otherwise accept.

  (2).  For a typical instance $\H$ of $\CSP(\Q_S)$, select some
  arbitrary vertex $h$ and then, for each possible interpretation of
  $h$ in $\Q_S$, test membership in $\CSP(\Q_{[k]\backslash\{i\}})$
  for each $i\in S$, ensuring $h$ is constrained to lie a consistent
  location in each.

  (3).  Consider some instance $\H$.  As above, we may assume that
  $\H$ is connected, balanced and is of sufficiently small height.  If
  $\H$ has a solution that does not involve the amalgamated vertex,
  then it has a solution within one of the paths $\Q_{S_i}$, in which
  case this can be discovered by systematically applying part (1) to
  each of the paths $\Q_{S_1},\Q_{S_2},\dots$.  This is a fixed finite
  number of applications of part (1).

  Now we consider solutions that do involve the amalgamated vertex,
  which we denote by $a$.  (In typical instances where the lemma is
  used, $a$ will be either an element of $A$ or an element of $R$
  and the $\Q_{S_i}$ will be the paths emanating from $a$.)  Without
  loss of generality, we will assume that $a$ is of height $0$ in the
  fan (rather than of height $n$).  For each vertex $h\in H$ we are
  going to first check that $h$ can be interpreted at height $0$ using
  the subroutine for checking height described earlier.  We now
  consider the components of the induced subgraph of $\H$ that avoids
  all vertices at height $0$ (given that $h$ has been set to height
  $0$).  Each such component $C$ is systematically tested for
  satisfaction in $\Q_{S_1}$ (with the lowest level vertices of $C$
  constrained to lie at height $1$ in $\Q_{S_1}$), then in $\Q_{S_2}$,
  and asking that it be constraint by the singleton relation $\{a\}$).
  If each component $C$ test positive for satisfaction in some
  $\Q_{S_i}$ (where $i$ possibly varies), then $\H$ can be satisfied
  in the fan.  Otherwise, replace $h$ by the next vertex of $\H$,
  until either a positive check is identified, or all vertices are
  shown to fail, and $\H$ is discovered as a NO instance.
\end{proof}

\subsection{Stage 1: Verification that $\G$ is balanced and a test for
  height.}\label{subsec:1} Let $\Pp_n$ denote the directed path on
vertices $0,1,\dots,n$: it is just $\Q_{\{1,\dots,k\}}$.  The CSP over
the directed path $\Pp_n$ on vertices is solvable in logspace, so as
an initial check, it is verified in logspace that $\G$ is balanced and
of height at most $n$, the height of $\mathcal{D}(\A)$.  Moreover, the
precise height of $\G$ can be recorded, by testing for membership in
$\CSP(\mathbb{P}_{n-1})$, $\CSP(\mathbb{P}_{n-2})$ and so on until
some minimal height is recorded.  Because singleton unary relations
can be added to $\mathbb{P}_k$ without changing the logspace
solvability of the CSP, a more refined test can be performed to
determine the minimal height of any chosen vertex $u$ (with heights
starting at height $0$): simply ask for satisfaction of $\G$ in
$\CSP(\mathbb{P}_{\hgt{\G}})$ with the vertex $u$ constrained by the
relation $\{0\}$, then $\{1\}$, and so on until some minimal height is
achieved.  This subroutine will be performed frequently (but without
further detailed discussion) throughout the remaining steps of the
argument.

\subsection{Stage 2: Elimination of ``short
  components''.} \label{subsec:2} If $\G$ contains some component of
height strictly less than $n$, then we can test directly whether or
not it is a YES or NO instance of $\CSP(\mathcal{D}(\A))$ (this is
explained in the next paragraph).  If any are NO instances, then so is
$\G$ and we can output some fixed NO instance of $\CSP(\A)$.
Otherwise (if all are YES instances), we may simply ignore these short
components.  If $\G$ itself has height less than $n$, then instead of
ignoring all components of $\G$ we simply output some fixed YES
instance of $\CSP(\A)$, completing the reduction.

The process for testing membership in $\CSP(\mathcal{D}(\A))$ is as
follows.  We consider some connected component $\H$.  Now in any
satisfying interpretation of $\H$ in $\mathcal{D}(\A)$, we must either
interpret within some single path $\Q_S$ connecting $A$ to $R$ in
$\mathcal{D}(\A)$, or at some fan of such paths emanating from some
vertex in $A$ or some vertex in $R$.  There are a fixed finite
number of such subgraphs of $\mathcal{D}(\A)$, and we may use Lemma
\ref{lem:pathassign} for each one.

For the remainder of the algorithm we will assume that \emph{all
  connected components have height $n$.}

\subsection{Stage 3: $\hgt\G=n$ (and all components of $\G$ have
  height $n$).}\label{subsec:3} In this case we eventually output an
actual structure $\B$ with the property that $\G$ is a YES instance of
$\CSP(\mathcal{D}(\A))$ if and only if $\B$ is a YES instance of
$\CSP(\A)$.  In fact we focus on the production of a preliminary
construction $\B'$ that is not specifically a relational structure,
but holds all the information for constructing $\B$ using undirected
graph reachability checks.  The output $\B'$ will consist of a list of
``generalised hyperedges'', that will (in step 3B) eventually become
the actual hyperedges of $\B$.  A typical generalised hyperedge will
consist of a labelled $k$ tuple $[s_1,s_2,\dots,s_k]_e$, where each
$s_i$ is a set of either existing vertices of $\G$, or some newly
constructed vertices, and where the ``label'' $e$ is itself a vertex
from $\G$ used for book-keeping and later amalgamation of the sets of
vertices.  (Some generalised hyperedges will not require the label
$e$.)

For the remainder of the argument, an \emph{internal component} of
$\G$ means a connected component of the induced subgraph of $\G$
obtained by removing all vertices of height $0$ and $n$.  Note that we
have already described that testing for height can be done in
logspace.  A \emph{base vertex} for such a component $C$ is a vertex
at height $0$ that is adjacent to $C$, and a \emph{top vertex} for $C$
is a level $n$ vertex adjacent to $C$.  Note that an internal
component may have none, or more than one base vertices, and similarly
for top vertices.  Every internal component must have at least one of
a base vertex or a top vertex however, because we have already
considered the case of ``short'' components in Stage 2.

We will frequently apply Lemma \ref{lem:pathassign}(1) to internal
components $C$, to discover the smallest $S$ for which $C$ is
satisfiable in $\Q_S$.  Once an internal component has a base or top
(and it must have at least one), then in a satisfying interpretation
of $\G$ in $\mathcal{D}(\A)$, the component $C$ must either be
satisfied within some single path, with any vertices adjacent to the
base of $C$ (or to a top of $C$) being interpreted adjacent to the
initial point of the path (or adjacent to the terminal point of the
path, respectively).  Thus we will tacitly assume that our
applications of Lemma \ref{lem:pathassign}(1) involve suitably
constraining $C$ to be consistent with this.  Also note that any
digraph of height $n$ is satisfiable $\Q_{[k]}$, so we will always
find some smallest $S$ using our test.

Connecting paths in $\mathcal{D}(\A)$ encode positions of base level
vertices in hyperedges, and we may use a check of Lemma
\ref{lem:pathassign}(1) to determine which positions are being
asserted as ``filled'' by any given internal component $C$ (recalling
that $C$ must either have a base vertex or a top vertex, so that
interpretations in paths $\Q_{[k]\backslash\{i\}}$ occur at fixed
positions only).  In this way, an internal component determines a
subset $\Gamma_C$ of $[k]$ by
\[
i\in \Gamma_C\text{ if and only if }C\text{ is satisfiable in }\Q_{[k]\backslash\{i\}}
\]
(Recall, that if $C$ has a base vertex, then we consider only
interpretations that would place the base vertex at
$\iota\Q_{[k]\backslash\{i\}}$, while if $C$ has a top vertex, then we
consider only interpretations that would place the base vertex at
$\tau\Q_{[k]\backslash\{i\}}$.)

\subsection{Stage 3A: Constructing the approximation $\B'$ to $\B$.}\label{subsec:3A}
To being with we do not output $\B$ itself, but rather some
approximation $\B'$ to $\B$.  This piece of information consists of a
list of ``generalised hyperedges'' plus a list of equalities.  These
generalised hyperedges consist of $k$-tuples of lists of vertex names:
vertices in the same list will later be identified to create $\B$, but
this is a separate construction.  Some hyperedges also encode some
extra vertex of $\G$ from which they were created.  So a typical
generalised hyperedge may look like $[V_1,V_2,\dots,V_k]_e$, where $e$
is some vertex of $\G$ (at height $n$) used for book-keeping purposes
and the $V_i$ are lists consisting of some vertices of $\G$ (from
height $0$) and some new vertices we create during the algorithm.
Other hyperedges may not require the special book-keeping subscript.

Note that any \emph{new} vertices created during the algorithm should
be different each time (even though we often use $x$ to denote such a
vertex): we should use some counter on a fixed spare piece of tape for
the entire algorithm; this counter is incremented at each creation of
a new variable, and its value recorded within the new vertex name.
(There will be only polynomially many new variables created, so only
logspace used to store this one counter.)
\begin{enumerate}
\item[1] To output the generalised hyperedges.  There are two causes
  for writing generalised hyperedges to the output: the first is due
  to vertices at height $n$, and the second is due to vertices at
  height $0$ that are the base vertex for some internal component with
  no top vertices.  The generalised hyperedges will be written in such
  a way to record some extra information that will be used for
  identifications.
\item[] For each vertex $e$ at height $n$ we will need to output a
  generalised hyperedge, however there may be many different vertices
  placed at a given position: these vertices will later be identified.
  We will also record in the encoding that the generalised hyperedge
  is created from vertex $e$.  The following process is performed for
  each height $n$ vertex $e$ and in each case, we perform the
  following process for $i=1$ to $k$.
\item[1.1] Systematically search for an internal component $C$ in
  which $i\in\Gamma(C)$ and for which $e$ is a top vertex.  These
  searches involve the following: we systematically search through all
  vertices of $\G$ until some $u$ is found to be undirected reachable
  from $e$ amongst vertices not at height $0$ or $n$.  To avoid
  unnecessary duplication, we may also check that $u$ does not lie in
  the same internal component as some earlier vertex (in which case we
  may ignore $u$: this internal component has already been
  considered).  Then we proceed to systematically search through all
  vertices of $\G$ to identify the internal component $C_u$ of $\G$
  containing $u$.  This component is then checked using
  Lemma~\ref{lem:pathassign} for whether $i\in \Gamma(C_u)$.  If $i\in
  \Gamma(C)$ we go to substep 1.1.1.  If $i\notin \Gamma(C_u)$ we
  increment $u$ and continue our search for an internal component $C$
  with $e$ as top and with $i\in \Gamma(C)$.  If no such components
  are encountered we proceed to substep~1.1.2.
\begin{itemize}
\item[1.1.1] We have identified an internal component $C$ with
  $i\in\Gamma(C)$ and for which $e$ is a top vertex.  If $C$ has base
  vertices $b_1,\dots,b_j$ then these will be written to the vertex
  set for the $i^{\rm th}$ coordinate of the output hyperedge.  If $C$
  does not have base vertices, then we will create some \emph{new}
  vertex $x$ and write the vertex set $\{x\}$ to the $i^{\rm th}$
  coordinate.
\item[1.1.2] No internal component $C$ is found with $i\in\Gamma(C)$
  and for which $e$ is a top vertex.  In this case only one vertex
  will appear in the vertex set for coordinate $i$ of this generalised
  hyperedge: a \emph{new} vertex $x$.
\end{itemize}
\item[1.2] Generalised hyperedges may also be created because of level
  $0$ vertices.  The following is performed for each level $0$ vertex
  $b$ and for each internal component $C$ for which $b$ is the base
  vertex and such that $C$ has no top vertex.  (If none are found
  there is nothing to do and no generalised hyperedge is written at
  step 1.2 for $b$.)  We create a generalised hyperedge by performing
  the following checks for $i=1,\dots,k$.
\begin{itemize}
\item[1.2.1] If $i\in\Gamma(C)$ then $\{b\}$ is placed in position $i$
  of the generalised hyperedge,
\item[1.2.2] If $i\notin\Gamma(C)$ then a new vertex $x$ is created
  and $\{x\}$ is placed in position~$i$ of the generalised hyperedge.
\end{itemize}
\item[2] Finally we output information that will later be used to find
  certain vertices that be forced to be identified in any satisfying
  interpretation of $\G$.
\begin{itemize}
\item[2.1] For each pair of distinct height $n$ vertices $e,f$, if $e$
  and $f$ are the top vertex for the same internal component, then we
  write $e=f$ to the output tape.
\item[2.2] For each pair of distinct height $0$ vertices $b,c$, if $b$
  and $c$ are base vertices for the same internal component we write
  $b=c$.
\end{itemize}
\end{enumerate}
This completes the construction of $\B'$.  There are clearly further
identifications that will be forced: for example, if $b$ appears in
the list of position $i$ vertices for some generalised hyperedge $e$,
and $c$ appears in the list of position $i$ vertices for some
generalised hyperedge $f$, and if $e=f$ has been output, then we must
have $b$ and $c$ identified.  Accounting for these is stage 3B.
\subsection{Stage 3B: construction of $\B$.}\label{subsec:3B}
We now need to construct $\B$ from the list of generalised hyperedges
and equalities.  We describe this as a separate logspace process, and
use the fact that a composition of logspace constructions is itself
logspace.  The actual vertices of $\B$ will consist of sets of the
vertices currently stated.  If desired, this could be simplified as a
later separate logspace process (perhaps by using only the earliest
vertex from each set).  Currently the input consists of generalised
hyperedges where the entry in a given position is a set of vertices of
$\G$ or new vertices.  To create $\B$ we only need to amalgamate these
sets, also taking into account the equality
constraints.  

In the following, a ``vertex'' refers to an element of some set within
the position of some hyperedge.  A ``vertex set'' consists of a set of
vertices.  The actual vertices of $\B$ will be vertex sets, produced
from those appearing within $\B'$ by amalgamation.

The amalgamation process involves considering an undirected graph on
the vertices, which we refer to as the \emph{equality graph}.  The
undirected edges of the equality graph arise in several different
ways.  There will be an undirected edge from a vertex $a$ to a vertex
$b$ if $a$ and $b$ lie within the same vertex set somewhere in the
input list, or if $a=b$ is written as an equality constraint.  For a
further kind reason for an undirected edge, recall that a hyperedge
created from a height $n$ vertex $e$ records the vertex $e$ in its
description.  The role of this is just so that $e$ acts as a place
holder, and we now use this.  There will be an undirected edge from
$a$ and $b$ if they appear in vertex sets at position $i$ of two
generalised hyperedges, either with the same label $e$, or with
different labels $e,f$ but where $e=f$ appears in the input.  There is
a logspace check for these undirected edges, and using the logspace
solvability of undirected graph reachability we may determine if two
vertices identified within our list of generalised hyperedges are
connected in the equality graph, all within logspace.

For each vertex $u$ we first check if there is some lexicographically
earlier vertex $v$ for which $u$ and $v$ are connected in the equality
graph.  If an earlier vertex is discovered, then we ignore $u$ and
continue to the next vertex.  Otherwise, if no earlier vertex is
discovered, we proceed to write down the vertex set of the component
of the equality graph containing $u$. For each $v$ lexicographically
later than $u$, we check whether $v$ is reachable from $u$ (in the
equality graph) and if so include it in vertex set of $u$.  For the
actual hyperedges of $\B$ we may simply write the existing generalised
hyperedges (removing the book-keeping subscript), which can be read in
the following way.  A vertex set $U$ appears in the $i^{\rm th}$
position of a hyperedge $E$ if the intersection of $U$ with the
vertices listed for position $i$ in $E$ is nontrivial.  Some
hyperedges may be repeated in this output and obviously this could
also be neatened by following with a totally new logspace reduction
(even to an adjacency matrix).

\subsection{If and only if.}
Any homomorphism $\phi$ from $\B$ (the amalgamated ``vertex sets'')
into $\A$, determines a function $\Phi$ from the height $0$ vertices
of $\G$ to the height $0$ vertices of $\mathcal{D}(\A)$.  The
construction of the hyperedges of $\B$ exactly reflects the
satisfiability of the internal components of $\G$, so that the
function $\Phi$ extends to cover all of $\G$ (this ignores any ``short
components'' that we considered directly in stage 2A: but they cannot
have been NO instances, as otherwise $\B$ was already created in stage
2 to be some fixed NO instance).

The converse is also true.  First, we only grouped vertices together
in vertex sets and their later amalgamation if they were forced to be
identified in any possible interpretation in $\mathcal{D}(\A)$.  So
any homomorphism $f$ from $\G$ to $\mathcal{D}(\A)$ determines a
function $F$ from the ``vertices'' of $\B$ to the vertices of $\A$.
The hyperedges of $\B$ were determined by the internal components of
$\G$, which $f$ is satisfying within the encoding (in
$\mathcal{D}(\A)$) of the hyperedges of $\A$.  So the $\B$ hyperedges
are preserved by $F$.

\subsection{An example}\label{subsec:example}
The following diagram depicts a reasonably general instance $\G$ of
$\CSP(\mathcal{D}(\A))$ in the case that $\A$ itself is a digraph, so
that $k=2$.  We are considering stage 3, so that $\G$ is a single
connected digraph of height $4$.  The vertices at height $0$ are
$b_1,\dots,b_6$, and the vertices at height $4$ are $e_1,\dots,e_4$.
The shaded regions depict internal components: each is labelled by a
subset of $\{1,2\}$, depicting $\Gamma(C)$.
\begin{center}
  \includegraphics[scale=0.85]{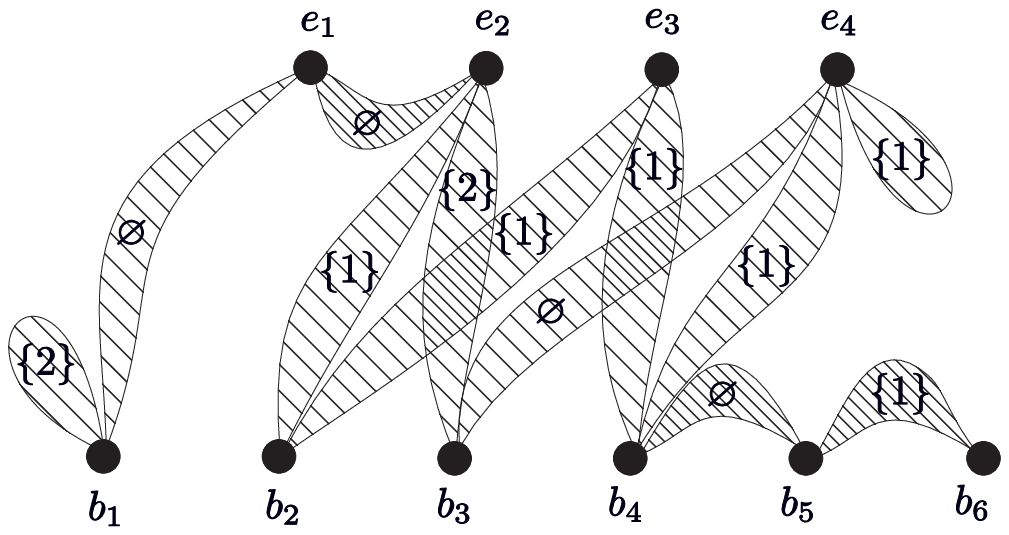}
\end{center}

Let us examine how Stage 3A proceeds.  We arrive at the first height
$4$ vertex $e_1$.  For $i=1$, discover no internal components with
$e_1$ as the top, and with $1\in \Gamma(C)$ (both have
$\Gamma(C)=\varnothing$, so we be in case 1.1.2) and therefore return
$\{x_1\}$ for the vertex set in the coordinate $1$.  For $i=2$, we
have the same outcome, so the edge that is generalised hyperedge that
is actually written is $[\{x_1\},\{x_2\}]_{e_1}$.

Then we proceed to the next height $4$ vertex $e_2$.  We encounter
just one internal component $C$ with $1\in C$, and its base vertices
are $\{b_2\}$ (so this is in case 1.1.1).  For $i=2$ we we also find
just one internal component whose $\Gamma$ value contains $2$, and it
has $\{b_3\}$ as the base vertices (also case 1.1.1).  The generalised
hyperedge $[\{b_2\},\{b_3\}]_{e_2}$ is written.

For $e_3$ and $i=1$ we encounter two internal components producing
base vertices.  We find $b_2$ as the only base vertex of the first,
and $b_4$ for the second (case 1.1.1), so the first coordinate of the
generalised hyperedge is $\{b_2,b_4\}$.  For $i=2$, no internal
components yield a base vertex (case 1.1.2), so we output $\{x_3\}$.
The actual generalised hyperedge written is
$[\{b_2,b_4\},\{x_3\}]_{e_2}$.

The vertex $e_4$ similar results in the generalised hyperedge
$[\{b_4,x_4\},\{x_5\}]_{e_4}$.

This completes step $1.1$ and we continue with step $1.2$.  We
discover the height $0$ vertex $b_1$ as the base of an internal
component $C$ with no top.  We find $1\notin \Gamma(C)$, so $\{x_6\}$
is written to the first coordinate of a generalised hyperedge (step
1.2.2 for $i=1$).  For $i=2$ we find $2\in\Gamma(C)$ so return
$\{b_1\}$ for the second coordinate.  The actual output written is
$[\{x_6\},\{b_1\}]$ (there are no subscripts to hyperedges from step
$1.2$).  Level $0$ vertices $b_4$ and $b_5$ also lead to the creation
of generalised hyperedges.  The overall output after the completion of
steps 1.1 and~1.2 is
\begin{align*}
&[\{x_1\},\{x_2\}]_{e_1}\qquad&\text{(from $e_1$, step 1.1)}\\
&[\{b_2\},\{b_3\}]_{e_2}\qquad&\text{(from $e_2$, step 1.1)}\\
&[\{b_2,b_4\},\{x_3\}]_{e_3}\qquad&\text{(from $e_3$, step 1.1)}\\
&[\{b_4,x_4\},\{x_5\}]_{e_4}\qquad&\text{(from $e_4$, step 1.1)}\\
&[\{x_6\},\{b_1\}]\qquad&\text{(from $b_1$, step 1.2)}\\
&[\{x_7\},\{x_8\}]\qquad&\text{(from $b_4$, step 1.2)}\\
&[\{x_9\},\{x_{10}\}]\qquad&\text{(from $b_5$, step 1.2)}\\
&[\{b_5\},\{x_{11}\}]\qquad&\text{(from $b_5$, step 1.2)}\\
&[\{b_6\},\{x_{12}\}]\qquad&\text{(from $b_6$, step 1.2)}
\end{align*}

For step 2 of the algorithm, we output the following equalities
\begin{align*}
&e_1=e_2,\qquad\text{(from 2.1)}\\
&b_4=b_5,b_5=b_6 \qquad\text{(from 2.2)}
\end{align*}
This completes Stage 3A: the list just given is $\B'$.  We note that
hyperedges such as $[\{x_7\},\{x_8\}]$ will be no hinderance to
satisfiability of $\B$ in $\A$, and we could word our algorithm to
avoid writing these altogether.

Stage 3B then produces the digraph $\B$ with hyperedges
\begin{align*}
&[\{b_2,b_4,b_5,b_6,x_1,x_4\},\{b_3,x_2\}]\\
&[\{b_2,b_4,b_5,b_6,x_1,x_4\},\{b_3,x_2\}]\\
&[\{b_2,b_4,b_5,b_6,x_1,x_4\},\{x_3\}]\\
&[\{b_2,b_4,b_5,b_6,x_1,x_4\},\{x_5\}]\\
&[\{x_6\},\{b_1\}]\\
&[\{x_7\},\{x_8\}]\\
&[\{x_9\},\{x_{10}\}]\\
&[\{b_2,b_4,b_5,b_6,x_1,x_4\},\{x_{11}\}]\\
&[\{b_2,b_4,b_5,b_6,x_1,x_4\},\{x_{12}\}]
\end{align*}
Which is a digraph with 12 vertices (namely, the 12 different sets of
vertices appearing in hyperedges).

The algorithm itself is the composite of stage 3A and stage 3B: we
have used the fact that a composition of logspace reductions is
logspace.

\section{Preserving Maltsev conditions}

Given a finite relational structure $\mathbb A$, we are interested in the following question: How similar are the algebras of polymorphisms of $\mathbb A$ and $\mathcal D(\mathbb A)$? More precisely, which equational properties (or \emph{Maltsev conditions}) do they share? In this section we give a partial answer to this question.


\subsection{The result}

We start by an overview and statement of the main result of this section. All the new notions are introduced later, in Subsection \ref{subsection:preliminaries}.

Since $\mathbb A$ is pp-definable from the digraph $\mathcal D (\mathbb A)$ (see Lemma \ref{lem:forwardreduction}),  it follows that $A$ and $R$ are subuniverses of $\mathcal D(\mathbb A)$ and for any $f\in\Pol\mathcal D (\mathbb A)$, the restriction $f|_A$ is a polymorphism of $\mathbb A$. Consequently, for any set of identities $\Sigma$, 
$$
\mathcal D (\mathbb A)\models\Sigma\text{ implies that }\mathbb A\models\Sigma.
$$

Theorem \ref{thm:1} (or, more precisely, Lemma \ref{lem:preserves_WNUs}) already shows that the above implication is in fact an equivalence, if $\Sigma$ is a Maltsev condition describing having a Taylor term, a WNU operation, or generating a congruence meet semidistributive variety. The list of such conditions, for which there is an equivalence, can be greatly expanded.

\begin{theorem}\label{thm:preserved_conditions}
Let $\mathbb A$ be a finite relational structure. Let $\Sigma$ be a linear idempotent set of identities such that the algebra of polymorphisms of the zigzag satisfies $\Sigma$ and each identity in $\Sigma$ is either balanced or contains at most two variables. Then 
$$
\mathcal D(\mathbb A)\models\Sigma\text{ if and only if }\mathbb A\models\Sigma.
$$
\end{theorem} 

The following corollary lists some popular properties that can be expressed as sets of identities satisfying the above assumptions. Indeed, they include many commonly encountered Maltsev conditions.

\begin{corollary}\label{cor:preserved_conditions}
Let $\mathbb A$ be a finite relational structure.  Then each of the following hold equivalently on $\mathbb A$ and $\mathcal D(\mathbb A)$.
\begin{enumerate}[(1)]
\item Taylor polymorphism or equivalently weak near-unanimity (WNU)
  polymorphism \cite{maroti-mckenzie} or equivalently cyclic
  polymorphism~\cite{b-k} \up(conjectured to be equivalent to being
  tractable if $\mathbb{A}$ is a core~\cite{b-j-k}\up)\up;
\item Polymorphisms witnessing $\text{SD}(\wedge)$
  \up(equivalent to bounded width~\cite{b-k2}\up)\up;
\item \up(for $k\geq 4$\up) $k$-ary edge polymorphism
  \up(equivalent to few subpowers \cite{BIMMVW}, \cite{IMMVW}\up)\up;
\item $k$-ary near-unanimity polymorphism \up(equivalent to strict
  width~\cite{fedvar}\up)\up;
\item totally symmetric idempotent (TSI) polymorphisms of all arities
  \up(equivalent to width $1$~\cite{width1}, \cite{fedvar}\up)\up;
\item Hobby-McKenzie polymorphisms \up(equivalent to the
  corresponding variety satisfying a non-trivial congruence lattice
  identity\up)\up;
\item Gumm polymorphisms witnessing congruence modularity;
\item J\'onsson polymorphisms witnessing congruence distributivity;
\item polymorphisms witnessing $\text{SD}(\vee)$ \up(conjectured to be
  equivalent to \up{\texttt{NL}}~\cite{LT:conj}\up);
\item \up(\up for $n\geq 3$\up) polymorphisms witnessing congruence
  $n$-permutability \up(together with \up{(9)} is conjectured to be equivalent
  to \up{\texttt{L}}~\cite{LT:conj}\up).
\end{enumerate}
\end{corollary}

Note that the list includes all six conditions for omitting types in
the sense of Tame Congruence Theory~\cite{hobbymckenzie}.


We will prove Theorem \ref{thm:preserved_conditions} and Corollary \ref{cor:preserved_conditions} in subsection \ref{subsection:proofs}.

\subsection{Preliminaries} \label{subsection:preliminaries}


Given a finite relational structure $\mathbb A$, let $\Pol\mathbb A$ denote the set of all polymorphisms of $\mathbb A$. The \emph{algebra of polymorphisms} of $\mathbb A$ is simply the algebra with the same universe whose operations are all polymorphisms of $\mathbb A$. A subset $B\subseteq A$ is a \emph{subuniverse} of $\mathbb A$, denoted by $B\leq\mathbb A$, if it is a subuniverse of the algebra of polymorphisms of $\mathbb A$, i.e., it is closed under all $f\in\Pol\mathbb A$.

An \emph{(operational) signature} is a (possibly infinite) set of operation symbols with arities assigned to them. By an \emph{identity} we mean an expression $u\approx v$ where $u$ and $v$ are terms in some signature. An identity $u\approx v$ is \emph{linear} if both $u$ and $v$ involve at most one occurrence of an operation symbol (e.g. $f(x,y)\approx g(x)$, or $h(x,y,x)\approx x$); and \emph{balanced} if the sets of variables occuring in $u$ and in $v$ are the same (e.g. $f(x,x,y)\approx g(y,x,x)$). 

A set of identities $\Sigma$ is \emph{linear} if it contains only
linear identities; \emph{balanced} if all the identities in $\Sigma$
are balanced; and \emph{idempotent} if for each operation symbol $f$
appearing in an identity of $\Sigma$, the identity
$f(x,x,\dots,x)\approx x$ is in $\Sigma$. \footnote{We can relax this
  condition and require the identity $f(x,x,\dots,x)\approx x$ only to
  be a \emph{consequence} of identities in $\Sigma$.} For example, the
identities $p(y,x,x)\approx y,\ p(x,x,y)\approx y,\ p(x,x,x)\approx x$
(defining the so called \emph{Maltsev} term) form a linear idempotent
set of identities which is not balanced.

The \emph{strong Maltsev condition}, a notion usual in universal
algebra, can be defined in this context as a finite set of
identities. A \emph{Maltsev condition} is an increasing chain of
strong Maltsev conditions, ordered by syntactical consequence. In all
results from this section, ``set of identities'' can be replaced with
``Maltsev condition''.

Let $\Sigma$ be a set of identities in a signature with operation
symbols $\mathcal{F}=\{ f_\lambda\mid\lambda\in\Lambda\}$. We say that
a relational structure $\mathbb{A}$ \emph{satisfies} $\Sigma$ (and
write \emph{$\mathbb A\models\Sigma$}), if for every
$\lambda\in\Lambda$ there is a polymorphism
$f^\mathbb{A}_\lambda\in\Pol\mathbb A$ such that the identities in
$\Sigma$ hold universally in $\mathbb A$ when for each
$\lambda\in\Lambda$ the symbol $f_\lambda$ is interpreted as
$f^\mathbb A_\lambda$.

\subsection{Polymorphisms of the zigzag}

In the following, let $\mathbb Z$ be a zigzag with vertices $00$, $01$, $10$ and $11$, i.e., the oriented path  \mbox{$00\boldsymbol\rightarrow 01\boldsymbol\leftarrow 10\boldsymbol\rightarrow 11$}. The digraph $\mathbb Z$ satisfies most of the important Maltsev conditions (an exception being the Maltsev term). We need the following.

\begin{lemma} \label{lem:Zdistrlat} 
The digraph $\mathbb Z$ satisfies any set of identities which holds in the variety of distributive lattices.
\end{lemma}
\begin{proof}
  Define the operations $\wedge$ and $\vee$ in the following way. Let
  $x\wedge y$ by the vertex from $\{x,y\}$ closer to $00$ and $x\vee
  y$ the vertex closer to $11$. These two operations are polymorphisms
  and they form a distributive lattice; the rest follows immediately.
\end{proof}

\begin{corollary}  \label{cor:ZdistrlatCorollary}
The digraph $\mathbb Z$ has a majority polymorphism, and it satisfies any balanced set of identities.
\end{corollary}
\begin{proof}
The ternary operation defined by $m(x,y,z)=(x\wedge y)\vee(x\wedge z)\vee(y\wedge z)$ (the \emph{median}) is a majority polymorphism.

Let $\Sigma$ be a balanced set of identities. For every operation symbol $f$ (say $k$-ary) occurring in $\Sigma$, we define $f^\mathbb Z(x_1,\dots,x_k)=\bigwedge_{i=1}^k x_i$. It is easy to check that $f^\mathbb Z$ is a polymorphism and that such a construction satisfies any balanced identity.
\end{proof}



\begin{lemma} \label{lem:Z3perm}
The digraph $\mathbb Z$ is congruence 3-permutable.
\end{lemma}
\begin{proof}
The ternary polymorphisms $p$ and $q$ witnessing $3$-permutability can be defined as follows:
\begin{align*}
p_1(x,y,z)=&\begin{cases}
01 & \text{if $01\in\{ x,y,z\}$ and $y\neq z$},\\
10 & \text{if $10\in\{x,y,z\}$ and $01\notin\{x,y,z\}$ and $y\neq z$},\\
x & \text{otherwise},
\end{cases}\\
p_2(x,y,z)=&\begin{cases}
01 & \text{if $01\in\{ x,y,z\}$ and $x\neq y$},\\
10 & \text{if $10\in\{x,y,z\}$ and $01\notin\{x,y,z\}$ and $x\neq y$},\\
z & \text{if $x=y$}.\\
x & \text{otherwise}
\end{cases}
\end{align*}
If we have triples $(a_1,a_2,a_3)$ and $(b_1,b_2,b_3)$ such that
$a_i\rightarrow b_i$, for $i=1,2,3$, then
$\{a_1,a_2,a_3\}\subseteq\{00,10\}$ and
$\{b_1,b_2,b_3\}\subseteq\{01,11\}$. Clearly the result follows if
$p_1(a_1,a_2,a_3)=10$ since $p_1(b_1,b_2,b_3)\in\{01,11\}$. The only
case where $p_1(a_1,a_2,a_3)\neq 10$ is when $a_1=00$ and $a_2=a_3=10$
in which case $b_1=01$ and $b_2,b_3\in\{01,11\}$. In any case,
$p_1(b_1,b_2,b_3)=01$. A similar argument works for $p_2$, and hence
$p_1$ and $p_2$ are polymorphisms.

The identities $p_1(x,y,y)\approx x$ and $p_2(x,x,y)\approx y$ follow
directly from the definitions of $p_1$ and $p_2$. To prove the
equation $p_1(x,x,y)\approx p_2(x,y,y)$ we can assume that $x\neq
y$. If $01$ or $10$ are in $\{x,y\}$, then $p_1$ and $p_2$ agree. If
not, then $p_1(x,x,y)=p_2(x,y,y)=x$.
\end{proof}

\subsection{Proofs} \label{subsection:proofs} In this subsection we
prove Theorem \ref{thm:preserved_conditions} and Corollary
\ref{cor:preserved_conditions}. Fix a finite relational structure,
without loss of generality we can assume that $\mathbb A=(A;R)$, where
$R$ is a $k$-ary relation. Fix an arbitrary linear order $\preceq$ on
the set $E=A\times R$. We define the mapping $\epsilon:\mathcal
D(\mathbb A)\to E$ by setting $\epsilon(x)$ to be the
$\preceq$-minimal $e\in E$ such that $x\in\mathbb P_e$. 


We will use the following linear order $\sqsubseteq$ on the vertices of $\mathcal D(\mathbb A)$: Put $x\sqsubset y$ if either of the following is true:
\begin{itemize}
\item $\lvl(x)<\lvl(y)$,
\item $\lvl(x)=\lvl(y)$ and $\epsilon(x)\prec\epsilon(y)$,
\item $\lvl(x)=\lvl(y)$, $\epsilon(x)=\epsilon(y)$, and $x$ is closer to $\iota\mathbb P_{\epsilon(x)}$ than $y$.
\end{itemize}
We will need the following easy fact, we leave the verification to the reader.
\begin{observation} 
Let $C$ and $D$ be subsets of $\mathcal D(\mathbb A)\setminus(A\cup R)$ such that
\begin{itemize}
\item for every $x\in C$ there exists $y\in D$ such that $x\boldsymbol{\rightarrow}y$, and
\item for every $y\in D$ there exists $x\in C$ such that $x\boldsymbol{\rightarrow}y$.
\end{itemize}
If $c$ and $d$ are the $\sqsubseteq$-minimal elemets of $C$ and $D$, respectively, then $c\boldsymbol{\rightarrow}d$.
\end{observation}

Note that the subset $\{00,10\}$ of the zigzag $\mathbb Z$ is closed
under all polymorphisms of $\mathbb Z$ (as it is pp-definable using
the formula $(\exists y)(x\boldsymbol\rightarrow y)$, see Lemma
\ref{lem:ppdef_polymorphisms}). The same holds for $\{01,11\}$. We use
this fact in the construction below, namely in case 3.

\subsubsection*{Proof of Theorem \ref{thm:preserved_conditions}.}

Let $\Sigma$ be a set of identities in operation symbols
$\{f_\lambda:\lambda\in\Lambda\}$ satisfying the assumptions. Let
$\{f^\mathbb A_\lambda\mid\lambda\in\Lambda\}$ and $\{f^\mathbb
Z_\lambda\mid\lambda\in\Lambda\}$ be interpretations of the
operation symbols witnessing $\mathbb A\models\Sigma$ and
$\mathbb Z\models\Sigma$, respectively.

We will now define polymorphisms $\{f^{\mathcal D(\mathbb
  A)}_\lambda\mid\lambda\in\Lambda\}$ witnessing that $\mathcal
D(\mathbb A)\models\Sigma$. Fix $\lambda\in\Lambda$ and assume that
$f_\lambda$ is $m$-ary. We split the definition of $f^{\mathcal
  D(\mathbb A)}_\lambda$ into several cases and subcases. Let $\mathbf
c\in \mathcal D(\mathbb A)^m$.

\medskip

\noindent\textbf{Case 1.} $\mathbf c\in A^m\cup R^m$.

\smallskip

\noindent\underline{1a}\quad If $\mathbf c\in A^m$, we define $f^{\mathcal D(\mathbb A)}_\lambda(\mathbf c)=f^\mathbb A_\lambda(\mathbf c)$. 

\smallskip

\noindent\underline{1b}\quad If $\mathbf c\in R^m$, we define $f^{\mathcal D(\mathbb A)}_\lambda(\mathbf c)=(f^\mathbb A_\lambda)^{(k)}(\mathbf c)$.

\smallskip

\noindent\textbf{Case 2.} $\mathbf c\in\Delta_m\setminus(A^m\cup R^m)$.

\smallskip

\noindent Let $e_i=\epsilon(c_i)$ for $i\in[m]$ and $e=(f^\mathbb
A_\lambda)^{(k+1)}(e_1,\dots,e_m)$. Let $l\in[k]$ be minimal such that
$c_i\in\mathbb P_{e_i,l}$ for all $i\in[m]$. (Its existence is
guaranteed by Lemma \ref{lem:diagonalcomponent} (i).)

\smallskip

\noindent\underline{2a}\quad If $\mathbb P_{e,l}$ is a single edge,
then we define $f^{\mathcal D(\mathbb A)}_\lambda(\mathbf c)$ to be
the vertex from $\mathbb P_{e,l}$ having the same level as all the
$c_i$'s.

\smallskip 

\noindent If $\mathbb P_{e,l}$ is a zigzag, then at least one of the
$\mathbb P_{e_i,l}$'s is a zigzag as well. (This follows from the
construction of $\mathcal D(\mathbb A)$ and the fact that $f^{\mathcal
  D(\mathbb A)}_\lambda$ preserves $R$.) 
For every $i\in[m]$ such that $\mathbb P_{e_i,l}$ is a zigzag let
$\Phi_i:\mathbb P_{e_i,l}\to\mathbb Z$ be the (unique)
isomorphism. Let $\Phi$ denote the isomorphism from $\mathbb P_{e,l}$
to $\mathbb Z$.

\smallskip

\noindent\underline{2b}\quad If all of the $\mathbb P_{e_i,l}$'s are
zigzags, then the value of $f^{\mathcal D(\mathbb A)}_\lambda$ is
defined as follows:
$$
f^{\mathcal D(\mathbb A)}_\lambda(\mathbf c)=\Phi^{-1}(f^\mathbb Z_\lambda(\Phi_1(c_1),\dots,\Phi_m(c_m))).
$$

\smallskip

\noindent\underline{2c}\quad Else, we define $f^{\mathcal D(\mathbb
  A)}_\lambda(\mathbf c)$ to be the $\sqsubseteq$-minimal element from
the set
$$
\{\Phi^{-1}(\Phi_i(c_i))\mid\mathbb P_{e_i,l}\text{ is a zigzag}\}.
$$

\smallskip

\noindent\textbf{Case 3.} $\mathbf c\notin\Delta_m$.

\smallskip

\noindent\underline{3a}\quad If $|\{\lvl(c_i)\mid i\in[m]\}|=1$ and
$|\{\epsilon(c_i)\mid i\in[m]\}|=2$, say that
$\epsilon(c_i)\in\{e,e'\}$ for all $i\in[m]$ and $e\prec e'$
, then
let $\Phi:\{e,e'\}\to\{00,10\}$ be the bijection mapping $e$ to $00$
and $e'$ to $10$. We define $f^{\mathcal D(\mathbb A)}_\lambda(\mathbf
c)$ to be the $\sqsubseteq$-minimal element from the set
$$
\{c_i:\epsilon(c_i)=\Phi^{-1}(f^\mathbb Z_\lambda(\Phi(\epsilon(c_1)),\dots,\Phi(\epsilon(c_m)))\}.
$$

\smallskip

\noindent\underline{3b}\quad If $|\{\lvl(c_i)\mid i\in[m]\}|=2$, say
that $\lvl(c_i)\in\{l,l'\}$ for all $i\in[m]$ and $l<l'$, then let
$\Phi:\{l,l'\}\to\{00,10\}$ be the bijection mapping $l$ to $00$ and
$l'$ to $10$. We define $f^{\mathcal D(\mathbb A)}_\lambda(\mathbf c)$
to be the $\sqsubseteq$-minimal element from the set
$$
\{c_i:\lvl(c_i)=\Phi^{-1}(f^\mathbb Z_\lambda(\Phi(\lvl(c_1)),\dots,\Phi(\lvl(c_m)))\}.
$$

\smallskip

\noindent\underline{3c}\quad In all other cases we define $f^{\mathcal
  D(\mathbb A)}_\lambda(\mathbf c)$ to be the $\sqsubseteq$-minimal
element from the set $\{c_1,\dots,c_m\}$.

\medskip

We need to verify that the operations we constructed are polymorphisms
and that they satisfy all identities from $\Sigma$. We divide the
proof into three claims.

\begin{claim}
For every $\lambda\in\Lambda$, $f^{\mathcal D(\mathbb A)}_\lambda$ is a polymorphism of $\mathcal D(\mathbb A)$.
\end{claim}
\begin{proof}
  Let $\mathbf c\boldsymbol\rightarrow\mathbf d$ be an edge in
  $\mathcal D(\mathbb A)^m$. Note that $\mathbf c\in\Delta_m$ if and
  only if $\mathbf d\in\Delta_m$. The tuple $\mathbf c$ cannot fall
  under subcase \underline{1b} or under \underline{3a}, because these
  cases both prevent an outgoing edge from $\mathbf c$ (see Lemma
  \ref{lem:diagonalcomponent} (ii) for why this is true for \underline{3a}).

  We  first consider the situation where $\mathbf c$ falls under
  subcase \underline{1a} of the definition. Then $\mathbf d$ falls
  under case 2 and, moreover, $d_i=\iota\mathbb P_{e_i,1}$ for all
  $i\in[m]$. It is not hard to verify that $f^{\mathcal D(\mathbb
    A)}_\lambda(\mathbf d)=\iota\mathbb P_{e,1}$. (In subcase
  \underline{2b} we need the fact that $f^\mathbb Z_\lambda$ is
  idempotent.) Therefore $f^{\mathcal D(\mathbb A)}_\lambda(\mathbf
  c)=\iota\mathbb P_e\boldsymbol{\rightarrow}\iota\mathbb
  P_{e,1}=f^{\mathcal D(\mathbb A)}_\lambda(\mathbf d)$ and the
  polymorphism condition holds. The argument is similar when $\mathbf
  d$ falls under subcase \underline{1b} (and so $\mathbf c$ under case
  2).

  Consider now that $\mathbf c$ falls under case 2. Then $\mathbf d$
  falls either under subcase \underline{1b}, which was handled in the
  above paragraph, or also under case 2. The choice of $e_1,\dots,e_m$
  and $e$ is the same for both $\mathbf c$ and $\mathbf d$. By Lemma
  \ref{lem:diagonalcomponent} (i), there exists $l\in[k]$ such that
  $c_i,d_i\in\mathbb P_{e_i,l}$ for all $i\in[m]$.

  If the value of $l$ is also the same for both $\mathbf c$ and
  $\mathbf d$, then $f^{\mathcal D(\mathbb A)}_\lambda(\mathbf
  c)\boldsymbol{\rightarrow}f^{\mathcal D(\mathbb A)}_\lambda(\mathbf
  d)$ follows easily; in subcase \underline{2a} trivially, in
  \underline{2b} from the fact that $f^\mathbb Z_\lambda$ is a
  polymorphism of $\mathbb Z$ and in \underline{2c} from the
  observation about $\sqsubseteq$.

  It may be the case that this $l$ is not minimal for the tuple
  $\mathbf c$, that is, that $c_i\in\mathbb P_{e_i,l-1}$ for all
  $i\in[m]$. However, it then follows that $c_i=\tau\mathbb
  P_{e_i,l-1}=\iota\mathbb P_{e_i,l}$ and thus $f^{\mathcal D(\mathbb
    A)}_\lambda(\mathbf c)=\iota\mathbb P_{e_i,l}$ (again, using
  idempotency of $f^\mathbb Z_\lambda$ in subcase
  \underline{2b}). Knowing this allows for the same argument as in the
  above paragraph to apply.

  If $\mathbf c$ falls under one of the subcases \underline{3b} or
  \underline{3c}, then $\mathbf d$ falls under the same subcase. The
  fact that $f^{\mathcal D(\mathbb A)}_\lambda(\mathbf
  c)\boldsymbol{\rightarrow}f^{\mathcal D(\mathbb A)}_\lambda(\mathbf
  d)$ follows from the above observation about $\sqsubseteq$. (In
  subcase \underline{3b} our construction ``chooses'' either the lower
  or the higher level, and it is easy to see that this choice is the
  same for both $\mathbf c$ and $\mathbf d$.)
\end{proof}

\begin{claim}
  The $f^{\mathcal D(\mathbb A)}_\lambda$'s satisfy every balanced
  identity from $\Sigma$.
\end{claim}
\begin{proof}
  Let $f_\lambda(\mathbf u)\approx f_\mu(\mathbf v)\in\Sigma$ be a
  balanced identity in $s$ distinct variables $\{x_1,\dots,x_s\}$. Let
  $\mathcal E:\{x_1,\dots,x_s\}\to\mathcal D(\mathbb A)$ be some
  evaluation of the variables. Let $\mathbf u^\mathcal E$ and $\mathbf
  v^\mathcal E$ denote the corresponding evaluation of these tuples.

  Note that both $f^{\mathcal D(\mathbb A)}_\lambda(\mathbf u^\mathcal
  E)$ and $f_\mu(\mathbf v^\mathcal E)$ fall under the same subcase of
  the definition. The subcase to be applied depends only on the set of
  elements occuring in the input tuple, except for case two, where the
  choice of $e$ matters as well. However, since the identity
  $f_\lambda(\mathbf u)\approx f_\mu(\mathbf v)$ holds in $\mathbb A$,
  this $e$ is the same for both $\mathbf u^\mathcal E$ and $\mathbf
  v^\mathcal E$. Therefore, to verify that $f^{\mathcal D(\mathbb
    A)}_\lambda(\mathbf u^\mathcal E)=f^{\mathcal D(\mathbb
    A)}_\mu(\mathbf v^\mathcal E)$, it is enough to consider the
  individual subcases separately.

  In case 1 it follows immediately from the fact that the identity
  holds in $\mathbb A$. In case 2 it is easily seen that both
  $f^{\mathcal D(\mathbb A)}_\lambda(\mathbf u^\mathcal E)$ and
  $f_\mu(\mathbf v^\mathcal E)$ have the same level, and since the
  identity holds in $\mathbb A$, they also lie on the same path
  $\mathbb P_{e,l}$. To see that these two elements are equal, note
  that in subcase \underline{2a} it is trivial, in \underline{2b} it
  follows directly from the fact that the identity holds in $\mathbb
  Z$, and in \underline{2c} we use the fact that the identity is
  balanced: they are both $\sqsubseteq$-minimal element of the same
  set $\{\mathcal E(x_1),\dots,\mathcal E(x_s)\}$. Similar arguments
  can be used in case 3. In \underline{3a} both $f^{\mathcal D(\mathbb
    A)}_\lambda(\mathbf u^\mathcal E)$ and $f_\mu(\mathbf v^\mathcal
  E)$ are chosen from the set $\{e,e'\}$; they are the same since the
  identity holds in $\mathbb Z$. In \underline{3b} the level is the
  same for both of them (since the identity holds in $\mathbb Z$) and
  they are both $\sqsubseteq$-minimal element of the set of elements
  from $\{\mathcal E(x_1),\dots,\mathcal E(x_s)\}$ lying on that
  level. In \underline{3c} both are $\sqsubseteq$-minimal element of
  the set $\{\mathcal E(x_1),\dots,\mathcal E(x_s)\}$.
\end{proof}

\begin{claim}
  The $f^{\mathcal D(\mathbb A)}_\lambda$'s satisfy every identity
  from $\Sigma$ in at most two variables.
\end{claim}
\begin{proof}
  Balanced identities fall under the scope of the previous
  claim. Since $\Sigma$ is idempotent, we may without loss of
  generality consider only identities of the form $f_\lambda(\mathbf
  u)\approx x$, where $\mathbf u\in\{x,y\}^m$. Suppose that $x$ and
  $y$ evaluate to $c$ and $d$ in $\mathbb D_\mathbb A$, respectively,
  and let $\mathbf c\in\{c,d\}^n$ be the evaluation of $\mathbf u$. We
  want to prove that $f^{\mathcal D(\mathbb A)}_\lambda(\mathbf
  c)\approx c$.

  The tuple $\mathbf c$ cannot fall into subcase \underline{3c} of the
  definition of $f^{\mathcal D(\mathbb A)}_\lambda$. If it falls into
  case 1, the equality follows from the fact that the identity holds
  in $\mathbb A$ while in subcases \underline{3a} and \underline{3b}
  we use the fact that it holds in $\mathbb Z$.

  In case 2 it is easily seen that $f^{\mathcal D(\mathbb
    A)}_\lambda(\mathbf c)$ lies on the same path $\mathbb P_{e,l}$ as
  $c$ (using that the identity holds in $\mathbb A$) as well as on the
  same level of this path. In \underline{2a} it is trivial that
  $f^{\mathcal D(\mathbb A)}_\lambda(\mathbf c)=c$ while in
  \underline{2b} it follows from the fact that the identity holds in
  $\mathbb Z$. If $\mathbf c$ falls under subcase \underline{2c}, then
  $\mathbb P_{\epsilon(d),l}$ must be a single edge, and thus
  $f^{\mathcal D(\mathbb A)}_\lambda(\mathbf c)$ is defined to be the
  $\sqsubseteq$-minimal element of the one element set $\{c\}$.
\end{proof}

\subsubsection*{Proof of Corollary \ref{cor:preserved_conditions}.}
All items are expressible by linear idempotent sets of identities. In
all items except (5) they are in at most two variables, in item (5)
the defining identities are balanced. It remains to check that all
these conditions are satisfied in the zigzag, which follows from Lemma
\ref{lem:Z3perm} for item (10) and from Lemma
\ref{cor:ZdistrlatCorollary} for all other items.

\end{document}